\pdfoutput=1
\documentclass[journal]{IEEEtran}

\usepackage{ifpdf}
\usepackage{cite}

\ifCLASSINFOpdf
  \usepackage[pdftex]{graphicx}
\else
  \usepackage[dvips]{graphicx}
\fi

\usepackage{enumitem}
\usepackage[cmex10]{amsmath}
\usepackage{amsfonts}
\usepackage{amssymb}
\usepackage{bm}
\usepackage{amsthm}
\usepackage{xcolor}
\usepackage{tikz}
\usepackage{graphicx}
\usetikzlibrary{decorations.pathreplacing}

\ifCLASSOPTIONcompsoc
 \usepackage[position=t,subrefformat=parens,labelformat=parens，caption=false,font=normalsize,labelfont=sf,textfont=sf]{subfig}
\else
 \usepackage[position=t,subrefformat=parens,labelformat=parens,caption=false,font=footnotesize]{subfig}
\fi

\newtheorem{theorem}{Theorem}
\newtheorem{lemma}{Lemma}
\newtheorem{corollary}{Corollary}
\newtheorem{proposition}{Proposition}
\newtheorem{definition}{Definition}

\theoremstyle{remark}
\newtheorem{remark}{Remark}

\newcommand*{\doubleC}{\mathbb{C}}
\newcommand*{\doubleE}{\mathbb{E}}
\newcommand*{\doubleR}{\mathbb{R}}
\newcommand*{\doubleZ}{\mathbb{Z}}

\newcommand*{\scriptA}{\mathcal{A}}
\newcommand*{\scriptC}{\mathcal{C}}
\newcommand*{\scriptN}{\mathcal{N}}
\newcommand*{\scriptD}{\mathcal{D}}
\newcommand*{\scriptO}{\mathcal{O}}

\newcommand*{\bolda}{\bm{a}}
\newcommand*{\boldb}{\bm{b}}
\newcommand*{\boldc}{\bm{c}}

\newcommand*{\bolde}{\bm{e}}

\newcommand*{\boldh}{\bm{h}}
\newcommand*{\boldi}{\bm{i}}

\newcommand*{\boldn}{\bm{n}}

\newcommand*{\boldp}{\bm{p}}

\newcommand*{\boldr}{\bm{r}}
\newcommand*{\bolds}{\bm{s}}

\newcommand*{\boldx}{\bm{x}}
\newcommand*{\boldy}{\bm{y}}
\newcommand*{\boldz}{\bm{z}}

\newcommand*{\boldA}{\bm{A}}
\newcommand*{\boldB}{\bm{B}}
\newcommand*{\boldC}{\bm{C}}
\newcommand*{\boldD}{\bm{D}}

\newcommand*{\boldF}{\bm{F}}

\newcommand*{\boldH}{\bm{H}}
\newcommand*{\boldI}{\bm{I}}

\newcommand*{\boldK}{\bm{K}}

\newcommand*{\boldM}{\bm{M}}

\newcommand*{\boldP}{\bm{P}}
\newcommand*{\boldQ}{\bm{Q}}
\newcommand*{\boldR}{\bm{R}}
\newcommand*{\boldS}{\bm{S}}
\newcommand*{\boldT}{\bm{T}}

\newcommand*{\boldV}{\bm{V}}
\newcommand*{\boldW}{\bm{W}}
\newcommand*{\boldX}{\bm{X}}
\newcommand*{\boldY}{\bm{Y}}

\newcommand*{\boldGamma}{\bm{\Gamma}}
\newcommand*{\boldDelta}{\bm{\Delta}}

\newcommand*{\boldLambda}{\bm{\Lambda}}

\newcommand*{\boldSigma}{\bm{\Sigma}}

\newcommand*{\boldPhi}{\bm{\Phi}}
\newcommand*{\boldPsi}{\bm{\Psi}}
\newcommand*{\boldXi}{\bm{\Xi}}

\newcommand*{\boldalpha}{\bm{\alpha}}
\newcommand*{\boldbeta}{\bm{\beta}}

\newcommand*{\boldeta}{\bm{\eta}}

\newcommand*{\boldtheta}{\bm{\theta}}

\newcommand*{\boldxi}{\bm{\xi}}

\newcommand*{\vecm}{\operatorname{vec}} 
\newcommand*{\matm}{\operatorname{mat}} 
\newcommand*{\diagm}{\operatorname{diag}} 
\newcommand*{\trace}{\operatorname{tr}} 

\newcommand*{\Real}{\operatorname{\mathfrak{R}}}
\newcommand*{\Imag}{\operatorname{\mathfrak{I}}}

\newcommand*{\proj}[1]{{\bm{\Pi}}_{#1}}
\newcommand*{\projp}[1]{{{\bm{\Pi}}^\perp_{#1}}}

\newcommand*{\Es}{\bm{E}_\mathrm{s}}

\newcommand*{\Ls}{\bm{\Lambda}_\mathrm{s}}

\newcommand*{\Lsone}{\bm{\Lambda}_\mathrm{s1}}
\newcommand*{\Lstwo}{\bm{\Lambda}_\mathrm{s2}}

\newcommand*{\En}{\bm{E}_\mathrm{n}}

\newcommand*{\Enonet}{\tilde{\bm{E}}_\mathrm{n1}}

\newcommand*{\dEnone}{\Delta\bm{E}_\mathrm{n1}}
\newcommand*{\dEntwo}{\Delta\bm{E}_\mathrm{n2}}

\newcommand*{\Lnone}{\bm{\Lambda}_\mathrm{n1}}
\newcommand*{\Lntwo}{\bm{\Lambda}_\mathrm{n2}}

\newcommand*{\Lnonet}{\tilde{\bm{\Lambda}}_\mathrm{n1}}

\newcommand*{\dLnone}{\Delta\bm{\Lambda}_\mathrm{n1}}

\newcommand*{\Rv}{\bm{R}_\mathrm{v}}
\newcommand*{\Rvone}{\bm{R}_\mathrm{v1}}
\newcommand*{\Rvtwo}{\bm{R}_\mathrm{v2}}
\newcommand*{\Rvoneh}{\hat{\bm{R}}_\mathrm{v1}}
\newcommand*{\Rvtwoh}{\hat{\bm{R}}_\mathrm{v2}}

\newcommand*{\dRvone}{\Delta\bm{R}_\mathrm{v1}}
\newcommand*{\dRvoneH}{\Delta\bm{R}_\mathrm{v1}^H}
\newcommand*{\dRvtwo}{\Delta\bm{R}_\mathrm{v2}}

\newcommand*{\Rvonet}{\tilde{\bm{R}}_\mathrm{v1}}

\newcommand*{\drvone}{\Delta\bm{r}_\mathrm{v1}}

\newcommand*{\Av}{\bm{A}_\mathrm{v}}
\newcommand*{\Avb}{\bar{\bm{A}}_\mathrm{v}}

\newcommand*{\AvH}{{\bm{A}_\mathrm{v}^H}}
\newcommand*{\avk}{{\bm{a}_\mathrm{v}(\theta_k)}}

\newcommand*{\avkH}{{\bm{a}_\mathrm{v}^H(\theta_k)}}
\newcommand*{\Davk}{{\dot{\bm{a}}_\mathrm{v}(\theta_k)}}

\newcommand*{\DavkH}{{\dot{\bm{a}}_\mathrm{v}^H(\theta_k)}}
\newcommand*{\Ad}{\bm{A}_\mathrm{d}} 
\newcommand*{\DAd}{\dot{\bm{A}}_\mathrm{d}} 

\newcommand*{\noisevar}{{\sigma^2_\mathrm{n}}}
\newcommand*{\noisevarinv}{{\sigma^{-2}_\mathrm{n}}}
\newcommand*{\noisevarsq}{{\sigma^4_\mathrm{n}}}
\newcommand*{\noisevarsqrt}{\sigma_\mathrm{n}}
\newcommand*{\noisevarp}[1]{{\sigma^{#1}_\mathrm{n}}}

\newcommand*{\Mv}{M_\mathrm{v}}

\newcommand*{\TMv}{\bm{T}_{M_\mathrm{v}}}
\newcommand*{\boldzero}{\bm{0}}

\newcommand*{\CRB}{\mathrm{CRB}}
\newcommand*{\SNR}{\mathrm{SNR}}

\newcommand*{\MSEan}{\mathrm{MSE}_\mathrm{an}}
\newcommand*{\MSEem}{\mathrm{MSE}_\mathrm{em}}
\newcommand*{\FIM}{\mathrm{FIM}}
\newcommand*{\dB}{\mathrm{dB}}

\begin{document}
\title{Coarrays, MUSIC, and the Cram\'{e}r Rao Bound}

\author{Mianzhi~Wang,~\IEEEmembership{Student Member,~IEEE,}
        and~Arye~Nehorai,~\IEEEmembership{Life Fellow,~IEEE}
\thanks{This work was supported by the AFOSR under Grant FA9550-11-1-0210 and the ONR Grant N000141310050.}%
\thanks{M. Wang and A. Nehorai are with the Preston M. Green Department of Electrical and Systems Engineering, Washington University in St. Louis, St. Louis, MO 63130 USA (e-mail: mianzhi.wang@wustl.edu; nehorai@ese.wustl.edu)}
}

\markboth{}%
{}

\maketitle

\begin{abstract}
Sparse linear arrays, such as co-prime arrays and nested arrays, have the attractive capability of providing enhanced degrees of freedom. By exploiting the coarray structure, an augmented sample covariance matrix can be constructed and MUSIC (MUtiple SIgnal Classification) can be applied to identify more sources than the number of sensors. While such a MUSIC  algorithm works quite well, its performance has not been theoretically analyzed. In this paper, we derive a simplified asymptotic mean square error (MSE) expression for the MUSIC algorithm applied to the coarray model, which is applicable even if the source number exceeds the sensor number. We show that the directly augmented sample covariance matrix and the spatial smoothed sample covariance matrix yield the same asymptotic MSE for MUSIC. We also show that when there are more sources than the number of sensors, the MSE converges to a positive value instead of zero when the signal-to-noise ratio (SNR) goes to infinity. This finding explains the ``saturation'' behavior of the coarray-based MUSIC algorithms in the high SNR region observed in previous studies. Finally, we derive the Cram\'{e}r-Rao bound (CRB) for sparse linear arrays, and conduct a numerical study of the statistical efficiency of the coarray-based estimator. Experimental results verify theoretical derivations and reveal the complex efficiency pattern of coarray-based MUSIC algorithms.
\end{abstract}


%

\section{Introduction}

\IEEEPARstart{E}{stimating} source directions of arrivals (DOAs) using sensors arrays plays an important role in the field of array signal processing. For uniform linear arrays (ULA), it is widely known that traditional subspace-based methods, such as MUSIC, can resolve up to $N - 1$ uncorrelated sources with $N$ sensors~\cite{schmidt_multiple_1986, stoica_music_1989, huang_exact_1993}. However, for sparse linear arrays, such as minimal redundancy arrays (MRA) \cite{moffet_minimum-redundancy_1968}, it is possible to construct an augmented covariance matrix by exploiting the coarray structure. We can then apply MUSIC to the augmented covariance matrix, and up to $\scriptO(N^2)$ sources can be resolved with only $N$ sensors \cite{moffet_minimum-redundancy_1968}. 

Recently, the development of co-prime arrays \cite{pal_coprime_2011, tan_sparse_2014,tan_direction_2014,qin_generalized_2015} and nested arrays \cite{pal_nested_2010, han_wideband_2013, han_nested_2014}, has generated renewed interest in sparse linear arrays, and it remains to investigate the performance of these arrays. 
The performance of the MUSIC estimator and its variants (e.g., root-MUSIC \cite{friedlander_root-music_1993, pesavento_unitary_2000}) was thoroughly analyzed by Stoica et al. in \cite{stoica_music_1989}, \cite{stoica_music_1990} and \cite{stoica_performance_1990}. The same authors also derived the asymptotic MSE expression of the MUSIC estimator, and rigorously studied its statistical efficiency. In \cite{li_performance_1993}, Li et al. derived a unified MSE expression for common subspace-based estimators (e.g., MUSIC and ESPRIT~\cite{roy_esprit-estimation_1989}) via first-order perturbation analysis. However, these results are based on the physical array model and make use of the statistical properties of the original sample covariance matrix, which cannot be applied when the coarray model is utilized. In \cite{gorokhov_unified_1996}, Gorokhov et al. first derived a general MSE expression for the MUSIC algorithm applied to matrix-valued transforms of the sample covariance matrix. While this expression is applicable to coarray-based MUSIC, its explicit form is rather complicated, making it difficult to conduct analytical performance studies. Therefore, a simpler and more revealing MSE expression is desired.

In this paper, we first review the coarray signal model commonly used for sparse linear arrays. We investigate two common approaches to constructing the augmented sample covariance matrix, namely, the direct augmentation approach (DAA)~\cite{abramovich_detection-estimation_2001, liu_remarks_2015} and the spatial smoothing approach~\cite{pal_nested_2010}. We show that MUSIC yields the same asymptotic estimation error for both approaches. We are then able to derive an explicit MSE expression that is applicable to both approaches. Our MSE expression has a simpler form, which may facilitate the performance analysis of coarray-based MUSIC algorithms. We observe that the MSE of coarray-based MUSIC depends on both the physical array geometry and the coarray geometry. We show that, when there are more sources than the number of sensors, the MSE does not drop to zero even if the SNR approaches infinity, which agrees with the experimental results in previous studies. Next, we derive the CRB of DOAs that is applicable to sparse linear arrays. We notice that when there are more sources than the number of sensors, the CRB is strictly nonzero as the SNR goes to infinity, which is consistent with our observation on the MSE expression.
It should be mentioned that during the review process of this paper, Liu et al.\ and Koochakzadeh et al. also independently derived the CRB for sparse linear arrays in \cite{koochakzadeh_cram_2016,liu_cramerrao}. In this paper, we provide a more rigorous proof the CRB's limiting properties in high SNR regions. We also include various statistical efficiency analysis by utilizing our results on MSE, which is not present in \cite{koochakzadeh_cram_2016,liu_cramerrao}.
Finally, we verify our analytical MSE expression and analyze the statistical efficiency of different sparse linear arrays via numerical simulations. We we observe good agreement between the empirical MSE and the analytical MSE, as well as complex efficiency patterns of coarray-based MUSIC.

Throughout this paper, we make use of the following notations. Given a matrix $\boldA$, we use $\boldA^T$, $\boldA^H$, and $\boldA^*$ to denote the transpose, the Hermitian transpose, and the conjugate of $\boldA$, respectively. We use $A_{ij}$ to denote the $(i,j)$-th element of $\boldA$, and $\bolda_i$ to denote the $i$-th column of $\boldA$. If $\boldA$ is full column rank, we define its pseudo inverse as $\boldA^\dagger = (\boldA^H \boldA)^{-1} \boldA^H$. We also define the projection matrix onto the null space of $\boldA$ as $\projp{\boldA} = \boldI - \boldA\boldA^\dagger$. Let $\boldA = [\bolda_1\,\bolda_2\,\ldots\,\bolda_N] \in \doubleC^{M \times N}$, and we define the vectorization operation as $\vecm(\boldA) = [\bolda_1^T\,\bolda_2^T\,\ldots\,\bolda_N^T]^T$, and $\matm_{M, N}(\cdot)$ as its inverse operation. We use $\otimes$ and $\odot$ to denote the Kronecker product and the Khatri-Rao product (i.e., the column-wise Kronecker product), respectively. We denote by $\Real(\boldA)$ and $\Imag(\boldA)$ the real and the imaginary parts of $\boldA$. If $\boldA$ is a square matrix, we denote its trace by $\trace(\boldA)$. In addition, we use $\boldT_M$ to denote a $M \times M$ permutation matrix whose anti-diagonal elements are one, and whose remaining elements are zero. We say a complex vector $\boldz \in \doubleC^M$ is \emph{conjugate symmetric} if $\boldT_M \boldz = \boldz^*$. We also use $\bolde_i$ to denote the $i$-th natural base vector in Euclidean space. For instance, $\boldA \bolde_i$ yields the $i$-th column of $\boldA$, and $\bolde_i^T \boldA$ yields the $i$-th row of $\boldA$.

\section{The Coarray Signal Model}

We consider a linear sparse array consisting of $M$ sensors whose locations are given by $\scriptD = \{d_1, d_2, \ldots, d_M\}$. Each sensor location $d_i$ is chosen to be the integer multiple of the smallest distance between any two sensors, denoted by $d_0$. Therefore we can also represent the sensor locations using the integer set $\bar{\scriptD} = \{\bar{d}_1, \bar{d}_2, \ldots, \bar{d}_M\}$, where $\bar{d}_i = d_i/d_0$ for $i = 1,2,\ldots,M$. Without loss of generality, we assume that the first sensor is placed at the origin. We consider $K$ narrow-band sources $\theta_1, \theta_2, \ldots, \theta_K$ 
impinging on the array from the far field. Denoting $\lambda$ as the wavelength of the carrier frequency, we can express the steering vector for the $k$-th source as
\begin{equation}
    \label{eq:steering-vector-basic}
    \bolda(\theta_k) = \begin{bmatrix}
        1 & e^{j\bar{d}_2\phi_k} & \cdots & e^{j\bar{d}_M\phi_k}
    \end{bmatrix}^T,
\end{equation}
where $\phi_k = (2\pi d_0 \sin\theta_k)/\lambda$.
Hence the received signal vectors are given by
\begin{equation}
    \label{eq:recv-basic}
    \boldy(t) = \boldA(\theta) \boldx(t) + \boldn(t), t = 1,2,\ldots,N, 
\end{equation}
where $\boldA = [\bolda(\theta_1)\,\bolda(\theta_2)\,\ldots\,\bolda(\theta_K)]$ denotes the array steering matrix, $\boldx(t)$ denotes the source signal vector, $\boldn(t)$ denotes additive noise, and $N$ denotes the number of snapshots. In the following discussion, we make the following assumptions:
\begin{enumerate}[label=\textbf{A\arabic*}]
    \item
        \label{ass:a1-uc-source}
        The source signals follow the unconditional model \cite{stoica_performance_1990} and are uncorrelated white circularly-symmetric Gaussian.
    \item
        \label{ass:a2-distinct-doa}
        The source DOAs are distinct (i.e., $\theta_k \neq \theta_l\ \forall k \neq l$).
    \item
        \label{ass:a3-gaussian-noise}
        The additive noise is white circularly-symmetric Gaussian and uncorrelated from the sources.
    \item
        \label{ass:a4-uc-snapshot}
        The is no temporal correlation between each snapshot.
\end{enumerate}
Under \ref{ass:a1-uc-source}--\ref{ass:a4-uc-snapshot}, the sample covariance matrix is given by
\begin{equation}
    \label{eq:cov-baisc}
    \boldR = \boldA \boldP \boldA^H + \noisevar \boldI,
\end{equation}
where $\boldP = \diagm(p_1, p_2, \ldots, p_K)$ denotes the source covariance matrix, and $\noisevar$ denotes the variance of the additive noise. By vectorizing $\boldR$, we can obtain the following coarray model:
\begin{equation}
    \label{eq:coarray-full-model}
    \boldr = \Ad \boldp + \noisevar \boldi,
\end{equation}
where $\Ad = \boldA^* \odot \boldA$, $\boldp = [p_1, p_2, \ldots, p_K]^T$, and $\boldi = \vecm(\boldI)$.

It has been shown in \cite{pal_nested_2010} that $\Ad$ corresponds to the steering matrix of the coarray whose sensor locations are given by $\scriptD_\mathrm{co} = \{d_m - d_n|1 \leq m,n \leq M\}$. By carefully selecting rows of $(\boldA^* \odot \boldA)$, we can construct a new steering matrix representing a virtual ULA with enhanced degrees of freedom. Because $\scriptD_\mathrm{co}$ is symmetric, this virtual ULA is centered at the origin. The sensor locations of the virtual ULA are given by $[-\Mv+1, -\Mv+2, \ldots, 0, \ldots, \Mv-1]d_0$, where $\Mv$ is defined such that $2\Mv-1$ is the size of the virtual ULA. Fig.~\ref{fig:array} provides an illustrative example of the relationship between the physical array and the corresponding virtual ULA. The observation vector of the virtual ULA is given by
\begin{equation}
    \label{eq:coarray-ula-model}
    \boldz = \boldF\boldr 
    = \boldA_\mathrm{c} \boldp + \noisevar \boldF \boldi,
\end{equation}
where $\boldF$ is the coarray selection matrix, whose detailed definition is provided in Appendix~\ref{app:f-def}, and $\boldA_\mathrm{c}$ represents the steering matrix of the virtual array. The virtual ULA can be divided into $\Mv$ overlapping uniform subarrays of size $\Mv$. The output of the $i$-th subarray is given by $\boldz_i = \boldGamma_i \boldz$ for $i = 1,2,\ldots,\Mv$, where $\boldGamma_i = [\boldzero_{\Mv \times (i-1)}\,\boldI_{\Mv \times \Mv}\,\boldzero_{\Mv \times (\Mv-i)}]$ represents the selection matrix for the $i$-th subarray.

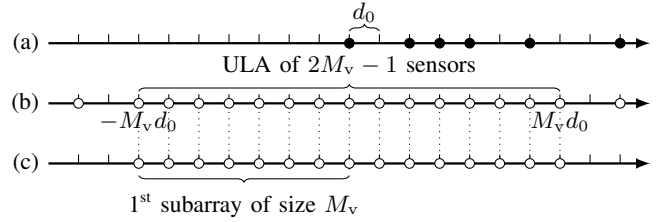
\begin{figure}
    \centering
    \begin{tikzpicture}[scale=0.8]
        \draw (0,3) node[left] {\small (a)};
        \draw[-latex, line width=0.3mm] (0,3) -- (10,3);
        \draw (5,2.9) -- (5,3.1);
        \foreach \x [evaluate=\x as \xscaled using \x/2+5] in
            {-9,-8,-7,-6,-5,-4,-3,-2,-1,0,1,2,3,4,5,6,7,8,9}
            \draw (\xscaled, 3) -- (\xscaled, 3.15);
        \foreach \x [evaluate=\x as \xscaled using \x/2+5] in {0,2,3,4,6,9}
            \draw[fill=black] (\xscaled,3) circle (0.08);
        \draw[decorate, decoration={brace}] (5,3.2) -- (5.5,3.2)
            node[midway, above] {\small $d_0$};

        \draw (0,2) node[left] {\small (b)};
        \draw[-latex, line width=0.3mm] (0,2) -- (10,2);
        \draw (5,1.9) -- (5,2.1);
        \foreach \x [evaluate=\x as \xscaled using \x/2+5] in
            {-9,-8,-7,-6,-5,-4,-3,-2,-1,0,1,2,3,4,5,6,7,8,9}
            \draw (\xscaled, 2) -- (\xscaled, 2.15);
        \foreach \x [evaluate=\x as \xscaled using \x/2+5] in
            {-9,-7,-6,-5,-4,-3,-2,-1,0,1,2,3,4,5,6,7,9}
            \draw[fill=white] (\xscaled,2) circle (0.08);

        \draw (0,1) node[left] {\small (c)};
        \draw[-latex, line width=0.3mm] (0,1) -- (10,1);
        \draw (5,0.9) -- (5,1.1);
        \foreach \x [evaluate=\x as \xscaled using \x/2+5] in
            {-9,-8,-7,-6,-5,-4,-3,-2,-1,0,1,2,3,4,5,6,7,8,9}
            \draw (\xscaled, 1) -- (\xscaled, 1.15);
        \foreach \x [evaluate=\x as \xscaled using \x/2+5] in
            {-7,-6,-5,-4,-3,-2,-1,0,1,2,3,4,5,6,7}
            \draw[fill=white] (\xscaled,1) circle (0.08);

        \foreach \x [evaluate=\x as \xscaled using \x/2+5] in
            {-7,-6,-5,-4,-3,-2,-1,0,1,2,3,4,5,6,7}
            \draw[dotted] (\xscaled, 1.1) -- (\xscaled, 1.9);

        \draw (1.5, 1.7) node {\small $-M_\mathrm{v}d_0$};
        \draw (8.5, 1.7) node {\small $M_\mathrm{v}d_0$};
        \draw[decorate, decoration={brace}]
            (1.5, 2.2) -- (8.5, 2.2)
            node[midway, above, yshift=0.1cm] {\small ULA of $2M_\mathrm{v}-1$ sensors};
        \draw[decorate, decoration={brace, mirror}]
            (1.5, 0.8) -- (5, 0.8)
            node[midway, below, yshift=-0.1cm] {\small 1\textsuperscript{st} subarray of size $M_\mathrm{v}$};
    \end{tikzpicture}
    \caption{A coprime array with sensors located at $[0, 2, 3, 4, 6, 9]\lambda/2$ and its coarray: (a) physical array; (b) coarray; (c) central ULA part of the coarray. }
    \label{fig:array}
\end{figure}

Given the outputs of the $\Mv$ subarrays, the augmented covariance matrix of the virtual array $\Rv$ is commonly constructed via one of the following methods~\cite{pal_nested_2010,liu_remarks_2015}:
\begin{subequations}
    \begin{align}
        \label{eq:def-rv1}
        \Rvone &= [\boldz_{\Mv}\,\boldz_{\Mv-1}\,\cdots\,\boldz_1], \\
        \label{eq:def-rv2}
        \Rvtwo &= \frac{1}{\Mv} \sum_{i=1}^{\Mv} \boldz_i \boldz_i^H,
    \end{align}
\end{subequations}
where method~\eqref{eq:def-rv1} corresponds to DAA , while method~\eqref{eq:def-rv2} corresponds to the spatial smoothing approach.

Following the results in \cite{pal_nested_2010} and \cite{liu_remarks_2015}, $\Rvone$ and $\Rvtwo$ are related via the following equality:
\begin{equation}
    \label{eq:rv1-rv2-relation}
    \Rvtwo
    = \frac{1}{\Mv}\Rvone^2
    = \frac{1}{\Mv}(\Av\boldP\Av^H + \noisevar\boldI)^2,
\end{equation}
where $\Av$ corresponds to the steering matrix of a ULA whose sensors are located at $[0, 1, \ldots, \Mv-1]d_0$. If we design a sparse linear array such that $\Mv > M$, we immediately gain enhanced degrees of freedom by applying MUSIC to either $\Rvone$ or $\Rvtwo$ instead of $\boldR$ in \eqref{eq:cov-baisc}. For example, in Fig.~\ref{fig:array}, we have a co-prime array with $\Mv = 8 > 6 = M$. Because MUSIC is applicable only when the number of sources is less than the number of sensors, we assume that $K < \Mv$ throughout the paper. This assumption, combined with \ref{ass:a2-distinct-doa}, ensures that $\Av$ is full column rank.

It should be noted that the elements in \eqref{eq:def-rv1} are obtained via linear operations on the elements in $\boldR$, and those in \eqref{eq:def-rv2} are obtained via quadratic operations. Therefore the statistical properties of $\Rvone$ and $\Rvtwo$ are different from that of $\boldR$. Consequently, traditional performance analysis for the MUSIC algorithm based on $\boldR$ cannot be applied to the coarray-based MUSIC. For brevity, we use the term direct augmentation based MUSIC (DA-MUSIC), and the term spatial smoothing based MUSIC (SS-MUSIC) to denote the MUSIC algorithm applied to $\Rvone$ and $\Rvtwo$, respectively. In the following section, we will derive a unified analytical MSE expression for both DA-MUSIC and SS-MUSIC.

\section{The MSE of coarray-based MUSIC}
In practice, the real sample covariance matrix $\boldR$ is unobtainable, and its maximum-likelihood estimate $\hat{\boldR} = 1/N\sum_{t=1}^N \boldx(t) \boldx(t)^H$ is used. Therefore $\boldz$, $\Rvone$, and $\Rvtwo$ are also replaced with their estimated versions $\hat{\boldz}$, $\Rvoneh$, and $\Rvtwoh$. Due to the estimation error $\Delta \boldR = \hat{\boldR} - \boldR$, the estimated noise eigenvectors will deviate from the true one, leading to DOA estimation errors. 

In general, the eigenvectors of a perturbed matrix are not well-determined \cite{stewart_error_1973}. For instance, in the very low SNR scenario, $\Delta \boldR$ may cause a subspace swap, and the estimated noise eigenvectors will deviate drastically from the true ones \cite{hawkes_performance_2001}. Nevertheless, as shown in \cite{li_performance_1993, gorokhov_unified_1996} and \cite{swindlehurst_performance_1992}, given enough samples and sufficient SNR, it is possible to obtain the closed-form expressions for DOA estimation errors via first-order analysis. Following similar ideas, we are able to derive the closed-form error expression for DA-MUSIC and SS-MUSIC, as stated in Theorem~\ref{thm:same-doa-err}.
\begin{theorem}
    \label{thm:same-doa-err}
    Let $\hat{\theta}_k^{(1)}$ and $\hat{\theta}_k^{(2)}$ denote the estimated values of the $k$-th DOA by DA-MUSIC and SS-MUSIC, respectively. Let $\Delta \boldr = \vecm(\hat{\boldR} - \boldR)$. Assume the signal subspace and the noise subspace are well-separated, so that $\Delta \boldr$ does not cause a subspace swap. Then
    \begin{equation}
        \label{eq:same-doa-err}
        \hat{\theta}_k^{(1)} - \theta_k
        \doteq \hat{\theta}_k^{(2)} - \theta_k
        \doteq -(\gamma_k p_k)^{-1} \Real(\boldxi_k^T \Delta\boldr),
    \end{equation}
    where $\doteq$ denotes asymptotic equality, and
    \begin{subequations}
        \begin{align}
            \label{eq:xi-k-def}
            \boldxi_k &= \boldF^T \boldGamma^T (\boldbeta_k \otimes \boldalpha_k), \\
            \label{eq:alpha-k-def}
            \boldalpha_k^T &= -\bolde_k^T \Av^\dagger, \\
            \label{eq:beta-k-def}
            \boldbeta_k &= \projp{\Av} \Davk, \\
            \label{eq:gamma-k-def}
            \gamma_k &= \DavkH \projp{\Av} \Davk, \\
            \label{eq:gamma-mat-def}
            \boldGamma &= [\boldGamma_{\Mv}^T\,\boldGamma_{\Mv-1}^T\,\cdots\boldGamma_1^T]^T, \\
            \label{eq:dav-def}
            \Davk &= \frac{\partial\avk}{\partial\theta_k}.
        \end{align}
    \end{subequations}
\end{theorem}
\begin{proof}
    See Appendix~\ref{app:thm-err-expression}.
\end{proof}

Theorem~\ref{thm:same-doa-err} can be reinforced by Proposition~\ref{prop:alpha-beta-xi-nonzero}. $\boldbeta_k \neq 0$ ensures that $\gamma_k^{-1}$ exists and \eqref{eq:same-doa-err} is well-defined, while $\boldxi_k \neq 0$ ensures that \eqref{eq:same-doa-err} depends on $\Delta \boldr$ and cannot be trivially zero.
\begin{proposition}
\label{prop:alpha-beta-xi-nonzero}
$\boldbeta_k, \boldxi_k \neq \boldzero$ for $k = 1,2,\ldots,K$.
\end{proposition}
\begin{proof}
We first show that $\boldbeta_k \neq \boldzero$ by contradiction. Assume $\boldbeta_k = \boldzero$. Then $\projp{\Av} \boldD \avk = \boldzero$, where $\boldD = \diagm(0,1,\ldots,\Mv-1)$. This implies that $\boldD \avk$ lies in the column space of $\Av$. Let $\boldh = e^{-j\phi_k} \boldD \avk$. We immediately obtain that $[\Av \ \boldh]$ is not full column rank. We now add $\Mv - K - 1$ distinct DOAs in $(-\pi/2, \pi/2)$ that are different from $\theta_1, \ldots, \theta_K$, and construct an extended steering matrix $\Avb$ of the $\Mv - 1$ distinct DOAs, $\theta_1, \ldots, \theta_{\Mv-1}$. Let $\boldB = [\Avb \ \boldh]$. It follows that $\boldB$ is also not full column rank. Because $\boldB$ is a square matrix, it is also not full row rank. Therefore there exists some non-zero $\boldc \in \doubleC^\Mv$ such that $\boldc^H \boldB = \boldzero$. Let $t_l = e^{j\phi_l}$ for $l = 1,2,\ldots,\Mv$, where $\phi_l = (2 \pi d_0 \sin\theta_k) / \lambda$. We can express $\boldB$ as
\begin{equation*}
    \begin{bmatrix}
        1 & 1 & \cdots & 1 & 0 \\
        t_1 & t_2 & \cdots & t_{\Mv-1} & 1 \\
        t_1^2 & t_2^2 & \cdots & t_{\Mv-1}^2 & 2t_k \\
        \vdots & \vdots & \ddots & \vdots & \vdots \\
        t_1^{\Mv-1} & t_2^{\Mv-1} &
        \cdots & t_{\Mv-1}^{\Mv-1} & (\Mv-1)t_k^{\Mv-2} \\
    \end{bmatrix}.
\end{equation*}
We define the complex polynomial $f(x) = \sum_{l=1}^{\Mv} c_l x^{l-1}$. It can be observed that $\boldc^T \boldB = \boldzero$ is equivalent to $f(t_l) = 0$ for $l = 1,2,\ldots,\Mv-1$, and $f'(t_k) = 0$. By construction, $\theta_l$ are distinct, so $t_l$ are $\Mv - 1$ different roots of $f(x)$. Because $\boldc \neq \boldzero$, $f(x)$ is not a constant-zero polynomial, and has at most $\Mv - 1$ roots. Therefore each root $t_l$ has a multiplicity of at most one. However, $f'(t_k) = 0$ implies that $t_k$ has a multiplicity of at least two, which contradicts the previous conclusion and completes the proof of $\boldbeta_k \neq \boldzero$.

We now show that $\boldxi_k \neq \boldzero$. By the definition of $\boldF$ in Appendix~\ref{app:f-def}, each row of $\boldF$ has at least one non-zero element, and each column of $\boldF$ has at most one non-zero element. Hence $\boldF^T \boldx = \boldzero$ for some $\boldx \in \doubleC^{2\Mv - 1}$ if and only of $\boldx = \boldzero$. It suffices to show that $\boldGamma^T (\boldbeta_k \otimes \boldalpha_k) \neq \boldzero$. By the definition of $\boldGamma$, we can rewrite $\boldGamma^T (\boldbeta_k \otimes \boldalpha_k)$ as $\tilde{\boldB}_k \boldalpha_k$, where
\begin{equation*}
    \tilde{\boldB}_k = 
    \begin{bmatrix}
        \beta_{k\Mv} 
            & 0 
            & \cdots
            & 0 \\
        \beta_{k(\Mv-1)}
            & \beta_{k\Mv}
            & \cdots
            & 0 \\
        \vdots
            & \vdots
            & \ddots
            & \vdots \\
        \beta_{k1}
            & \beta_{k2}
            & \cdots
            & \beta_{k\Mv} \\
        0
            & \beta_{k1}
            & \cdots
            & \beta_{k(\Mv-1)} \\
        \vdots
            & \vdots
            & \ddots
            & \vdots \\
        0
            & 0
            & \cdots
            & \beta_{k1}  
    \end{bmatrix},
\end{equation*}
and $\beta_{kl}$ is the $l$-th element of $\boldbeta_k$. Because $\boldbeta_k \neq \boldzero_k$ and $K < \Mv$, $\tilde{\boldB}_k$ is full column rank. By the definition of pseudo inverse, we know that $\boldalpha_k \neq \boldzero$. Therefore $\tilde{\boldB}_k \boldalpha_k \neq \boldzero$, which completes the proof of $\boldxi_k \neq \boldzero$.
\end{proof}

One important implication of Theorem~\ref{thm:same-doa-err} is that DA-MUSIC and SS-MUSIC share the same first-order error expression, despite the fact that $\Rvone$ is constructed from the second-order statistics, while $\Rvtwo$ is constructed from the fourth-order statistics. Theorem~\ref{thm:same-doa-err} enables a unified analysis of the MSEs of DA-MUSIC and SS-MUSIC, which we present in Theorem~\ref{thm:MSE-MUSIC}.

\begin{theorem}
    \label{thm:MSE-MUSIC}
    Under the same assumptions as in Theorem~\ref{thm:same-doa-err}, the asymptotic second-order statistics of the DOA estimation errors by DA-MUSIC and SS-MUSIC share the same form:
    \begin{equation}
        \label{eq:MSE-MUSIC-thm}
        \doubleE[(\hat{\theta}_{k_1} - \theta_{k_1})
            (\hat{\theta}_{k_2} - \theta_{k_2})]
        \doteq \frac{
            \Real[\boldxi_{k_1}^H (\boldR \otimes \boldR^T) \boldxi_{k_2}]
            }{N p_{k_1} p_{k_2} \gamma_{k_1}\gamma_{k_2}}.
    \end{equation}
\end{theorem}
\begin{proof}
See Appendix~\ref{app:thm-mse-music}.
\end{proof}

By Theorem~\ref{thm:MSE-MUSIC}, it is straightforward to write the unified asymptotic MSE expression as
\begin{equation}
    \label{eq:doa-mse}
    \epsilon(\theta_k)
    = \frac{\boldxi_k^H (\boldR \otimes \boldR^T) \boldxi_k}
        {N p_k^2 \gamma_k^2}.
\end{equation}

Therefore the MSE\footnote{For brevity, when we use the acronym ``MSE'' in the following discussion, we refer to the asymptotic MSE, $\epsilon(\theta_k)$, unless explicitly stated.} depends on both the physical array geometry and the coarray geometry. The physical array geometry is captured by $\boldA$, which appears in $\boldR \otimes \boldR^T$. The coarray geometry is captured by $\Av$, which appears in $\boldxi_k$ and $\gamma_k$. Therefore, even if two arrays share the same coarray geometry, they may not share the same MSE because their physical array geometry may be different.

It can be easily observed from \eqref{eq:doa-mse} that $\epsilon(\theta_k) \to 0$ as $N \to \infty$. However, because $p_k$ appears in both the denominator and numerator in \eqref{eq:doa-mse}, it is not obvious how the MSE varies with respect to the source power $p_k$ and noise power $\noisevar$. Let $\bar{p}_k = p_k / \noisevar$ denote the signal-to-noise ratio of the $k$-th source. Let $\bar{\boldP} = \diagm(\bar{p}_1, \bar{p}_2, \ldots, \bar{p}_K)$, and $\bar{\boldR} = \boldA \bar{\boldP} \boldA^H + \boldI$. We can then rewrite $\epsilon(\theta_k)$ as
\begin{equation}
    \label{eq:doa-mse-normalized}
    \epsilon(\theta_k)
    = \frac{
        \boldxi_k^H (\bar{\boldR} \otimes \bar{\boldR}^T) \boldxi_k
        }{N \bar{p}_k^2 \gamma_k^2}.
\end{equation}
Hence the MSE depends on the SNRs instead of the absolute values of $p_k$ or $\noisevar$. To provide an intuitive understanding how SNR affects the MSE, we consider the case when all sources have the same power. In this case, we show in Corollary~\ref{corr:mse-desc} that the MSE asymptotically decreases as the SNR increases.

\begin{corollary}
\label{corr:mse-desc}
Assume all sources have the same power $p$. Let $\bar{p} = p/\noisevar$ denote the common SNR. Given sufficiently large $N$, the MSE $\epsilon(\theta_k)$ decreases monotonically as $\bar{p}$ increases, and
\begin{equation}\label{eq:mse-infty-snr}
    \lim_{\bar{p} \to \infty}
    \epsilon(\theta_k)
    = \frac{1}{N\gamma_k^2}
     \| \boldxi_k^H (\boldA \otimes \boldA^*) \|_2^2.
\end{equation}
\end{corollary}
\begin{proof}
The limiting expression can be derived straightforwardly from \eqref{eq:doa-mse-normalized}. For monotonicity, without loss of generality, let $p = 1$, so $\bar{p} = 1/\noisevar$. Because $f(x) = 1/x$ is monotonically decreasing on $(0, \infty)$, it suffices to show that $\epsilon(\theta_k)$ increases monotonically as $\noisevar$ increases. Assume $0 < s_1 < s_2$, and we have
\begin{equation*}
    \left.\epsilon(\theta_k)\right|_{\noisevar = s_2} -
    \left.\epsilon(\theta_k)\right|_{\noisevar = s_1}
    = \frac{1}{N\gamma_k^2}
        \boldxi_k^H \boldQ \boldxi_k,
\end{equation*}
where $\boldQ = (s_2 - s_1)[(\boldA\boldA^H) \otimes \boldI + \boldI \otimes (\boldA\boldA^H) + (s_2 + s_1)\boldI]$. Because $\boldA\boldA^H$ is positive semidefinite, both $(\boldA\boldA^H) \otimes \boldI$ and $\boldI \otimes (\boldA\boldA^H)$ are positive semidefinite. Combined with our assumption that $0 < s_1 < s_2$, we conclude that $\boldQ$ is positive definite. By Proposition~\ref{prop:alpha-beta-xi-nonzero} we know that $\boldxi_k \neq \boldzero$. Therefore $\boldxi_k^H \boldQ \boldxi_k$ is strictly greater than zero, which implies the MSE monotonically increases as $\noisevar$ increases.
\end{proof}

Because both DA-MUSIC and SS-MUSIC work also in cases when the number of sources exceeds the number of sensors, we are particularly interested in their limiting performance in such cases. As shown in Corollary~\ref{corr:mse-g-zero}, when $K \geq M$, the corresponding MSE is strictly greater than zero, even though the SNR approaches infinity. This corollary explains the ``saturation'' behavior of SS-MUSIC in the high SNR region as observed in \cite{qin_generalized_2015} and \cite{pal_nested_2010}.
Another interesting implication of Corollary~\ref{corr:mse-g-zero} is that when $2 \leq K < M$, the limiting MSE is not necessarily zero. Recall that in \cite{stoica_music_1989}, it was shown that the MSE of the traditional MUSIC algorithm will converge to zero as SNR approaches infinity. We know that both DA-MUSIC and SS-MUSIC will be outperformed by traditional MUSIC in high SNR regions when $2 \leq K < M$. Therefore, we suggest using DA-MUSIC or SS-MUSIC only when $K \geq M$.
\begin{corollary}
\label{corr:mse-g-zero}
Following the same assumptions in Corollary~\ref{corr:mse-desc},
\begin{enumerate}
    \item When $K = 1$, $\lim_{\bar{p} \to \infty} \epsilon(\theta_k) = 0$;
    \item When $2 \leq K < M$, $\lim_{\bar{p} \to \infty} \epsilon(\theta_k) \geq 0$;
    \item When $K \geq M$, $\lim_{\bar{p} \to \infty} \epsilon(\theta_k) > 0$.
\end{enumerate}
\end{corollary}
\begin{proof}
The right-hand side of \eqref{eq:mse-infty-snr} can be expanded into
\begin{equation*}
    \frac{1}{N\gamma_k^2}
    \sum_{m=1}^K \sum_{n=1}^K \|\boldxi_k^H [\bolda(\theta_m)\otimes\bolda^*(\theta_n)]\|_2^2.
\end{equation*}
By the definition of $\boldF$, $\boldF [\bolda(\theta_m)\otimes\bolda^*(\theta_m)]$ becomes
\begin{equation*}
    [e^{j(\Mv-1)\phi_m}, e^{j(\Mv-2)\phi_m}, \ldots, e^{-j(\Mv-1)\phi_m}].
\end{equation*}
Hence $\boldGamma \boldF [\bolda(\theta_m)\otimes\bolda^*(\theta_m)]  = \bolda_{\mathrm{v}}(\theta_m) \otimes \bolda_{\mathrm{v}}^*(\theta_m)$. Observe that
\begin{equation*}
\begin{aligned}
    \boldxi_k^H [\bolda(\theta_m)\otimes\bolda^*(\theta_m)]
    =& (\boldbeta_k \otimes \boldalpha_k)^H 
        (\bolda_{\mathrm{v}}(\theta_m) \otimes \bolda_{\mathrm{v}}^*(\theta_m)) \\
    =& (\boldbeta_k^H \bolda_{\mathrm{v}}(\theta_m))
        (\boldalpha_k^H \bolda_{\mathrm{v}}^*(\theta_m)) \\
    =& (\dot{\bolda}_{\mathrm{v}}^H(\theta_k) \projp{\Av} 
            \bolda_{\mathrm{v}}(\theta_m))
        (\boldalpha_k^H \bolda_{\mathrm{v}}^*(\theta_m)) \\
    =& 0.
\end{aligned}
\end{equation*}
We can reduce the right-hand side of \eqref{eq:mse-infty-snr} into
\begin{equation*}
    \frac{1}{N\gamma_k^2}
    \sum_{\substack{1\leq m,n \leq K \\ m\neq n}} \|\boldxi_k^H [\bolda(\theta_m)\otimes\bolda^*(\theta_n)]\|_2^2.
\end{equation*}
Therefore when $K = 1$, the limiting expression is exactly zero. When $2 \leq K < M$, the limiting expression is not necessary zero because when $m \neq n$, $\boldxi_k^H [\bolda(\theta_m)\otimes\bolda^*(\theta_n)]$ is not necessarily zero.

When $K \geq M$, $\boldA$ is full row rank. Hence $\boldA \otimes \boldA^*$ is also full row rank. By Proposition~\ref{prop:alpha-beta-xi-nonzero} we know that $\boldxi_k \neq 0$, which implies that $\epsilon(\theta_k)$ is strictly greater than zero.
\end{proof}

\section{The Cram\'{e}r-Rao Bound}
The CRB for the unconditional model \eqref{eq:recv-basic} has been well studied in \cite{stoica_performance_1990}, but only when the number of sources is less than the number of sensors and no prior knowledge of $\boldP$ is given. For the coarray model, the number of sources can exceed the number of sensors, and $\boldP$ is assumed to be diagonal. Therefore, the CRB derived in \cite{stoica_performance_1990} cannot be directly applied. Based on \cite[Appendix 15C]{kay_fundamentals_1993}, we provide an alternative CRB based on the signal model \eqref{eq:recv-basic}, under assumptions \ref{ass:a1-uc-source}--\ref{ass:a4-uc-snapshot}.

For the signal model \eqref{eq:recv-basic}, the parameter vector is defined by
\begin{equation}
    \boldeta = [\theta_1, \ldots, \theta_K, p_1, \ldots, p_k, \noisevar]^T,
\end{equation}
and the $(m,n)$-th element of the Fisher information matrix (FIM) is given by \cite{kay_fundamentals_1993, stoica_performance_1990}
\begin{equation}
    \label{eq:FIM-element}
    \FIM_{mn} = N\trace\Bigg[
        \frac{\partial{\boldR}}{\partial{\eta_m}}
        \boldR^{-1}
        \frac{\partial{\boldR}}{\partial{\eta_n}}
        \boldR^{-1}
    \Bigg].
\end{equation}
Observe that $\trace(\boldA\boldB) = \vecm(\boldA^T)^T \vecm(\boldB)$, and that $\vecm(\boldA\boldX\boldB) = (\boldB^T \otimes \boldA)\vecm(\boldX)$. We can rewrite \eqref{eq:FIM-element} as
\begin{equation*}
    \begin{aligned}
        \FIM_{mn}
        &= N
        \Bigg[\frac{\partial{\boldr}}{\partial{\eta_m}}\Bigg]^H
        (\boldR^T \otimes \boldR)^{-1}
        \frac{\partial{\boldr}}{\partial{\eta_n}}.
    \end{aligned}
\end{equation*}
Denote the derivatives of $\boldr$ with respect to $\boldeta$ as
\begin{equation}
    \label{eq:pr-peta-def}
    \frac{\partial{\boldr}}{\partial{\boldeta}}
    = \Bigg[
        \frac{\partial{\boldr}}{\partial{\theta_1}}\,
        \cdots\,
        \frac{\partial{\boldr}}{\partial{\theta_K}}\,
        \frac{\partial{\boldr}}{\partial{p_1}}\,
        \cdots\,
        \frac{\partial{\boldr}}{\partial{p_K}}\,
        \frac{\partial{\boldr}}{\partial{\noisevar}}
    \Bigg].
\end{equation}
The FIM can be compactly expressed by
\begin{equation}
    \label{eq:fim-unpartitioned}
    \FIM =
    \Bigg[\frac{\partial{\boldr}}{\partial{\boldeta}}\Bigg]^H
    (\boldR^T \otimes \boldR)^{-1}
    \frac{\partial{\boldr}}{\partial{\boldeta}}.
\end{equation}
According to \eqref{eq:coarray-full-model}, we can compute the derivatives in \eqref{eq:pr-peta-def} and obtain
\begin{equation}
    \label{eq:pr-peta-final}
    \frac{\partial{\boldr}}{\partial{\boldeta}}
    = \begin{bmatrix}
        \DAd\boldP & \Ad & \boldi
    \end{bmatrix},
\end{equation}
where $\DAd = \dot{\boldA}^* \odot \boldA + \boldA^* \odot \dot{\boldA}$, $\Ad$ and $\boldi$ follow the same definitions as in \eqref{eq:coarray-full-model}, and
\begin{equation*}
    \dot{\boldA} = 
    \Bigg[
        \frac{\partial \bolda(\theta_1)}{\partial \theta_1} \,
        \frac{\partial \bolda(\theta_2)}{\partial \theta_2} \,
        \cdots \,
        \frac{\partial \bolda(\theta_K)}{\partial \theta_K}
    \Bigg].
\end{equation*}
Note that \eqref{eq:pr-peta-final} can be partitioned into two parts, specifically, the part corresponding to DOAs and the part corresponding to the source and noise powers. We can also partition the FIM. Because $\boldR$ is positive definite, $(\boldR^T \otimes \boldR)^{-1}$ is also positive definite, and its square root $(\boldR^T \otimes \boldR)^{-1/2}$ also exists. Let
\begin{equation*}
    \boldM_{\boldtheta} = (\boldR^T \otimes \boldR)^{-1/2}
    \DAd\boldP,
\end{equation*}
\begin{equation*}
    \boldM_{\bolds} = (\boldR^T \otimes \boldR)^{-1/2}
    \big[ \Ad\, \boldi \big].
\end{equation*}
We can write the partitioned FIM as
\begin{equation*}
    \FIM = N\begin{bmatrix}
        \boldM_{\boldtheta}^H \boldM_{\boldtheta} &
        \boldM_{\boldtheta}^H \boldM_{\bolds} \\
        \boldM_{\bolds}^H \boldM_{\boldtheta} &
        \boldM_{\bolds}^H \boldM_{\bolds}
    \end{bmatrix}.
\end{equation*}
The CRB matrix for the DOAs is then obtained by block-wise inversion:
\begin{equation}
    \label{eq:crb-final}
    \CRB_{\boldtheta} = \frac{1}{N}
    (\boldM_{\boldtheta}^H 
    \projp{\boldM_{\bolds}}
    \boldM_{\boldtheta})^{-1},
\end{equation}
where $\projp{\boldM_{\bolds}} = \boldI - \boldM_{\bolds} (\boldM_{\bolds}^H \boldM_{\bolds})^{-1} \boldM_{\bolds}^H$. It is worth noting that, unlike the classical CRB for the unconditional model introduced in \cite[Remark 1]{stoica_performance_1990}, expression \eqref{eq:crb-final} is applicable even if the number of sources exceeds the number of sensors. 

\begin{remark}
Similar to \eqref{eq:doa-mse}, $\CRB_{\boldtheta}$ depends on the SNRs instead of the absolute
values of $p_k$ or $\noisevar$. Let $\bar{p}_k = p_k / \noisevar$, and $\bar{\boldP} = \diagm(\bar{p}_1, \bar{p}_2, \ldots, \bar{p}_K)$. We have
\begin{align}
    \label{eq:m-theta-normalized}
    \boldM_{\boldtheta} &= (\bar{\boldR}^T \otimes \bar{\boldR})^{-1/2}
    \DAd\bar{\boldP}, \\
    \label{eq:m-s-normalized}
    \boldM_{\bolds} &= \noisevarinv(\bar{\boldR}^T \otimes \bar{\boldR})^{-1/2}
    \big[ \Ad\, \boldi \big].
\end{align}
Substituting \eqref{eq:m-theta-normalized} and \eqref{eq:m-s-normalized} into \eqref{eq:crb-final}, the term $\noisevar$ gets canceled, and the resulting $\CRB_{\boldtheta}$ depends on the ratios $\bar{p}_k$ instead of absolute values of $p_k$ or $\noisevar$.
\end{remark}

\begin{remark}
The invertibility of the FIM depends on the coarray structure. In the noisy case, $(\boldR^T \otimes \boldR)^{-1}$ is always full rank, so the FIM is invertible if and only if $\partial\boldr/\partial\boldeta$ is full column rank. By \eqref{eq:pr-peta-final} we know that the rank of $\partial\boldr/\partial\boldeta$ is closely related to $\Ad$, the coarray steering matrix. Therefore $\CRB_{\boldtheta}$ is not valid for an arbitrary number of sources, because $\Ad$ may not be full column rank when too many sources are present.
\end{remark}

\begin{proposition}
\label{prop:crb-snr-infty}
Assume all sources have the same power $p$, and $\partial\boldr / \partial\boldeta$ is full column rank. Let $\bar{p} = p / \noisevar$.
\begin{enumerate}[label=(\arabic*)]
    \item If $K < M$, and $\lim_{\bar{p} \to \infty} \CRB_{\boldtheta}$ exists, it is zero under mild conditions.
    \item If $K \geq M$, and $\lim_{\bar{p} \to \infty} \CRB_{\boldtheta}$ exists, it is positive definite.
\end{enumerate}
\end{proposition}
\begin{proof}
See Appendix~\ref{app:crb-snr-infty}.
\end{proof}

While infinite SNR is unachievable from a practical standpoint, Proposition~\ref{prop:crb-snr-infty} gives some useful theoretical implications. When $K < M$, the limiting MSE (13) in Corollary~\ref{corr:mse-desc} is not necessarily zero. However, Proposition~\ref{prop:crb-snr-infty} reveals that the CRB may approach zero when SNR goes to infinity. This observation implies that both DA-MUSIC and SS-MUSIC may have poor statistical efficiency in high SNR regions. When $K \geq M$, Proposition~\ref{prop:crb-snr-infty} implies that the CRB of each DOA will converge to a positive constant, which is consistent with Corollary~\ref{corr:mse-g-zero}.

\section{Numerical Analysis}
\label{sec:numerical-results}
In this section, we numerically analyze of DA-MUSIC and SS-MUSIC by utilizing \eqref{eq:doa-mse} and \eqref{eq:crb-final}. We first verify the MSE expression \eqref{eq:MSE-MUSIC-thm} introduced in Theorem~\ref{thm:MSE-MUSIC} through Monte Carlo simulations.
We then examine the application of \eqref{eq:same-doa-err} in predicting the resolvability of two closely placed sources, and analyze the asymptotic efficiency of both estimators from various aspects.
Finally, we investigate how the number of sensors affect the asymptotic MSE.

In all experiments, we define the signal-to-noise ratio (SNR) as 
\begin{equation*}
    \mathrm{SNR} = 10\log_{10}\frac{\min_{k=1,2,\ldots,K} p_k}{\noisevar},
\end{equation*}
where $K$ is the number of sources.

Throughout Section~\ref{subsec:verification}, \ref{subsec:res} and \ref{subsec:eff}, we consider the following three different types of linear arrays with the following sensor configurations:
\begin{itemize}
    \item Co-prime Array~\cite{pal_coprime_2011}: $[0,3,5,6,9,10,12,15,20,25]\lambda/2$
    \item Nested Array~\cite{pal_nested_2010}: $[1,2,3,4,5,10,15,20,25,30]\lambda/2$
    \item MRA~\cite{ishiguro_minimum_1980}: $[0,1,4,10,16,22,28,30,33,35]\lambda/2$
\end{itemize}
All three arrays share the same number of sensors, but difference apertures.

\subsection{Numerical Verification}\label{subsec:verification}
We first verify \eqref{eq:doa-mse} via numerical simulations. We consider 11 sources with equal power, evenly placed between $-67.50^\circ$ and $56.25^\circ$, which is more than the number of sensors. We compare the difference between the analytical MSE and the empirical MSE under different combinations of SNR and snapshot numbers. The analytical MSE is defined by
\begin{equation*}
    \mathrm{MSE}_\mathrm{an}
    = \frac{1}{K}\sum_{k=1}^K \epsilon(\theta_k),
\end{equation*}
and the empirical MSE is defined by
\begin{equation*}
    \mathrm{MSE}_\mathrm{em}
    = \frac{1}{KL}\sum_{l=1}^L\sum_{k=1}^K
        \big(\hat{\theta}_k^{(l)} - \theta_k^{(l)}\big)^2,
\end{equation*}
where $\theta_k^{(l)}$ is the $k$-th DOA in the $l$-th trial, and $\hat{\theta}_k^{(l)}$ is the corresponding estimate.

\begin{figure}[ht]
    \centering
    \includegraphics[width=0.95\linewidth]{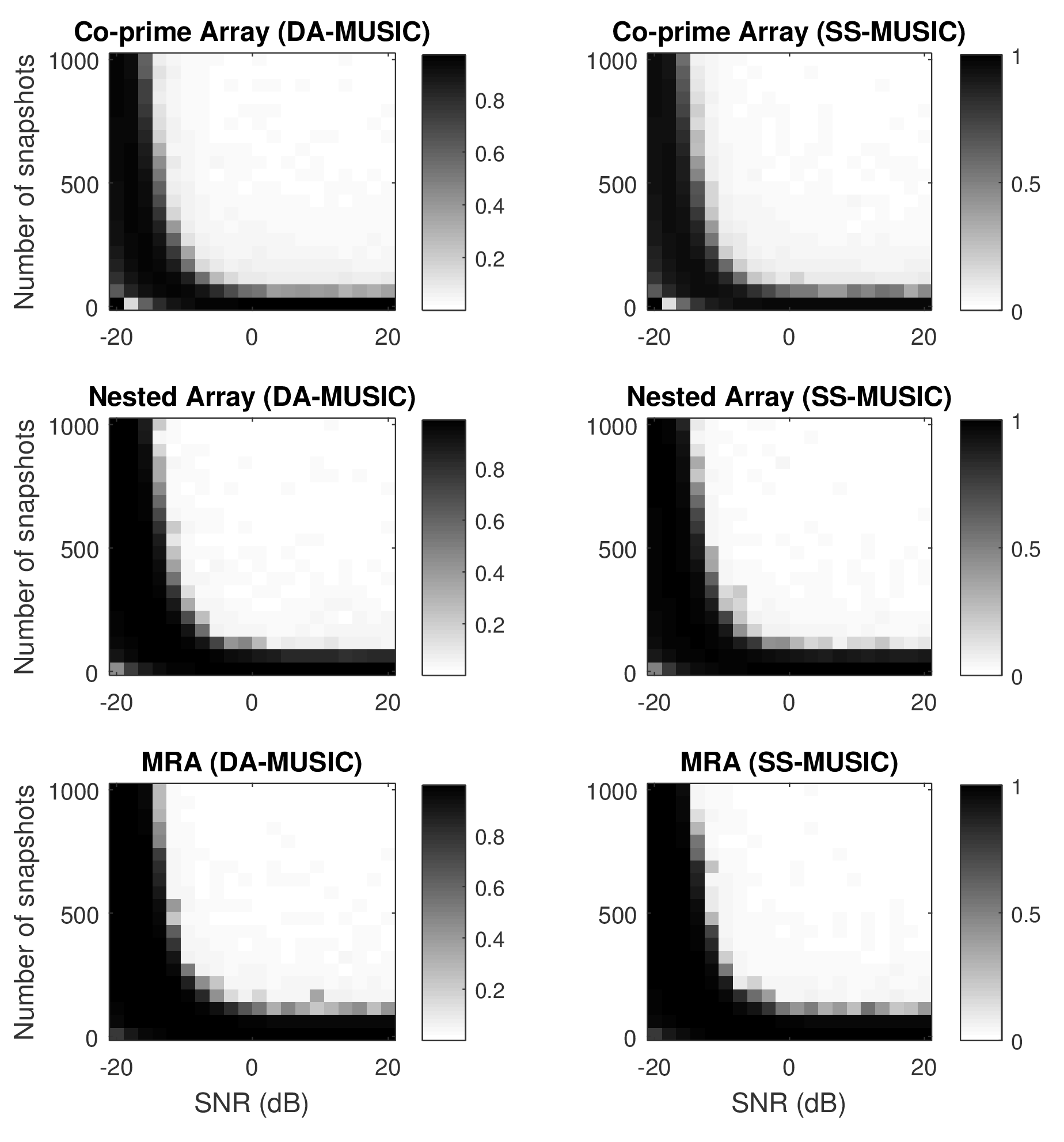}
    \caption{$|\MSEan - \MSEem|/\MSEem$ for different types of arrays under different numbers of snapshots and different SNRs.}
    \label{fig:acc_all}
\end{figure}

Fig.~\ref{fig:acc_all} illustrates the relative errors between $\MSEan$ and $\MSEem$ obtained from 10,000 trials under various scenarios. It can be observed that $\MSEem$ and $\MSEan$ agree very well given enough snapshots and a sufficiently high SNR. It should be noted that at 0dB SNR, \eqref{eq:same-doa-err} is quite accurate when 250 snapshots are available. In addition. there is no significant difference between the relative errors obtained from DA-MUSIC and those from SS-MUSIC. These observations are consistent with our assumptions, and verify Theorem~\ref{thm:same-doa-err} and Theorem~\ref{thm:MSE-MUSIC}. 

We observe that in some of the low SNR regions, $|\MSEan - \MSEem|/\MSEem$ appears to be smaller even if the number of snapshots is limited. In such regions, $\MSEem$ actually ``saturates'', and $\MSEan$ happens to be close to the saturated value. Therefore, this observation does not imply that \eqref{eq:doa-mse} is valid in such regions.

\subsection{Prediction of Resolvability}\label{subsec:res}
One direct application of Theorem~\ref{thm:MSE-MUSIC} is predicting the resolvability of two closely located sources. We consider two sources with equal power, located at $\theta_1 = 30^\circ - \Delta\theta/2$, and $\theta_2 = 30^\circ + \Delta\theta/2$, where $\Delta\theta$ varies from $0.3^\circ$ to $3.0^\circ$. We say the two sources are correctly resolved if the MUSIC algorithm is able to identify two sources, and the two estimated DOAs satisfy $|\hat{\theta}_i - \theta_i| < \Delta\theta/2$, for $i \in \{1,2\}$. The probability of resolution is computed from 500 trials. For all trials, the number of snapshots is fixed at 500, the SNR is set to 0dB, and SS-MUSIC is used.

For illustration purpose, we analytically predict the resolvability of the two sources via the following simple criterion:
\begin{equation}
    \label{eq:res-criterion}
    \epsilon(\theta_1) + \epsilon(\theta_2) 
    \underset{\mathrm{Resolvable}}{\overset{\mathrm{Unresovalble}}{\gtreqless}} \Delta\theta.
\end{equation}
Readers are directed to \cite{liu_statistical_2007} for a more comprehensive criterion.

Fig.~\ref{fig:res} illustrates the resolution performance of the three arrays under different $\Delta\theta$, as well as the thresholds predicted by \eqref{eq:res-criterion}. The MRA shows best resolution performance of the three arrays, which can be explained by the fact that the MRA has the largest aperture. The co-prime array, with the smallest aperture, shows the worst resolution performance. Despite the differences in resolution performance, the probability of resolution of each array drops to nearly zero at the predicted thresholds. This confirms that \eqref{eq:doa-mse} provides a convenient way of predicting the resolvability of two close sources.

\begin{figure}[ht]
    \centering
    \includegraphics[scale=0.6]{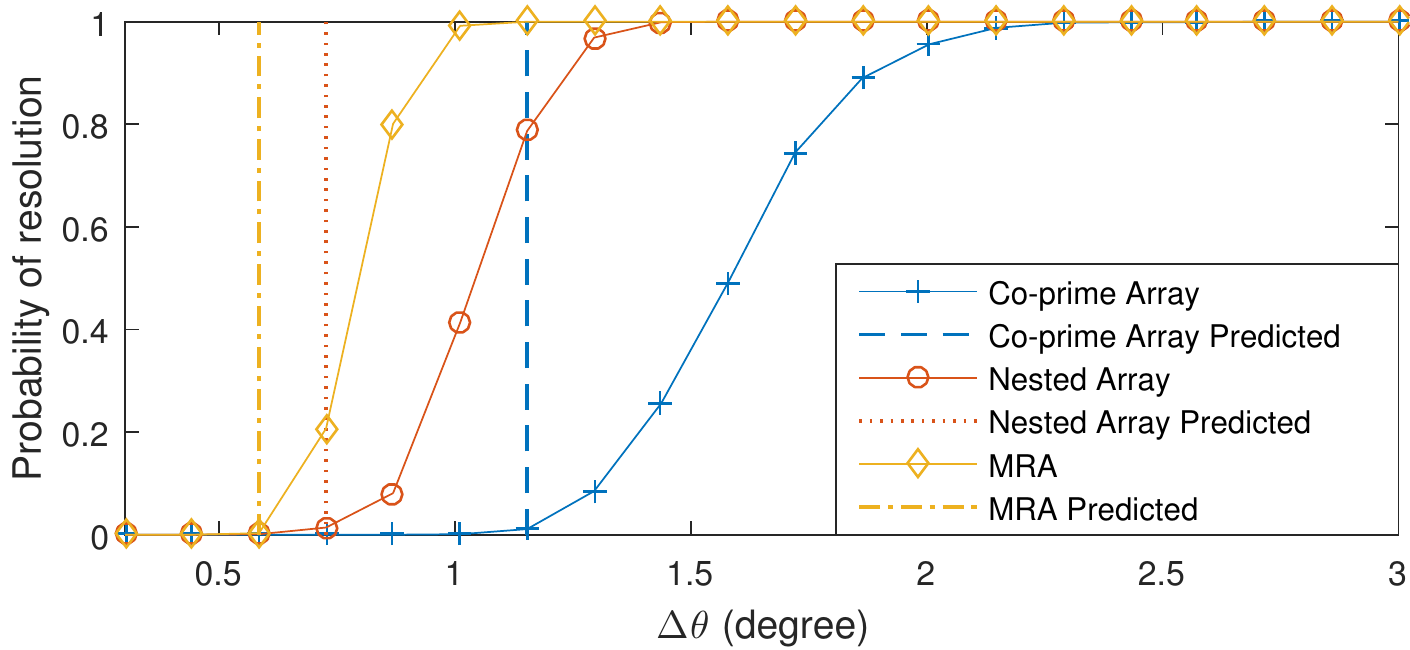}
    \caption{Probability of resolution vs. source separation, obtained from 500 trials. The number of snapshots is fixed at 500, and the SNR is set to 0dB.}
    \label{fig:res}
\end{figure}

\subsection{Asymptotic Efficiency Study}\label{subsec:eff}
In this section, we utilize \eqref{eq:doa-mse} and \eqref{eq:crb-final} to study the asymptotic statistical efficiency of DA-MUSIC and SS-MUSIC under different array geometries and parameter settings. We define their average efficiency as
\begin{equation}
    \label{eq:avg-eff}
    \kappa = \frac{\trace{\CRB_{\boldtheta}}}{\sum_{k=1}^K \epsilon(\theta_k)}.
\end{equation}
For efficient estimators we expect $\kappa = 1$, while for inefficient estimators we expect $0 \leq \kappa < 1$.

We first compare the $\kappa$ value under different SNRs for the three different arrays. We consider three cases: $K=1$, $K=6$, and $K = 12$. The $K$ sources are located at $\{-60^\circ + [120(k-1)/(K-1)]^\circ|k = 1,2,\ldots, K\}$, and all sources have the same power. As shown in Fig.~\subref*{fig:eff-1}, when only one source is present, $\kappa$ increases as the SNR increases for all three arrays. However, none of the arrays leads to efficient DOA estimation. Interestingly, despite being the least efficient geometry in the low SNR region, the co-prime array achieves higher efficiency than the nested array in the high SNR region.
When $K = 6$, we can observe in Fig.~\subref*{fig:eff-6} that $\kappa$ decreases to zero as SNR increases. This rather surprising behavior suggests that both DA-MUSIC and SS-MUSIC are not statistically efficient methods for DOA estimation when the number of sources is greater than one and less than the number of sensors. It is consistent with the implication of Proposition~\ref{prop:crb-snr-infty} when $K < M$.
When $K = 12$, the number of sources exceeds the number of sensors. We can observe in Fig.~\ref{fig:eff-12} that $\kappa$ also decreases as SNR increases. However, unlike the case when $K = 6$, $\kappa$ converges to a positive value instead of zero.

The above observations imply that DA-MUSIC and SS-MUSIC achieve higher degrees of freedom at the cost of decreased statistical efficiency. When statistical efficiency is concerned and the number of sources is less than the number of sensors, one might consider applying MUSIC directly to the original sample covariance $\boldR$ defined in \eqref{eq:cov-baisc} \cite{vaidyanathan_direct-music_2012}.

\begin{figure}[ht]
    \centering
    \subfloat[]{%
        \includegraphics[scale=0.54]{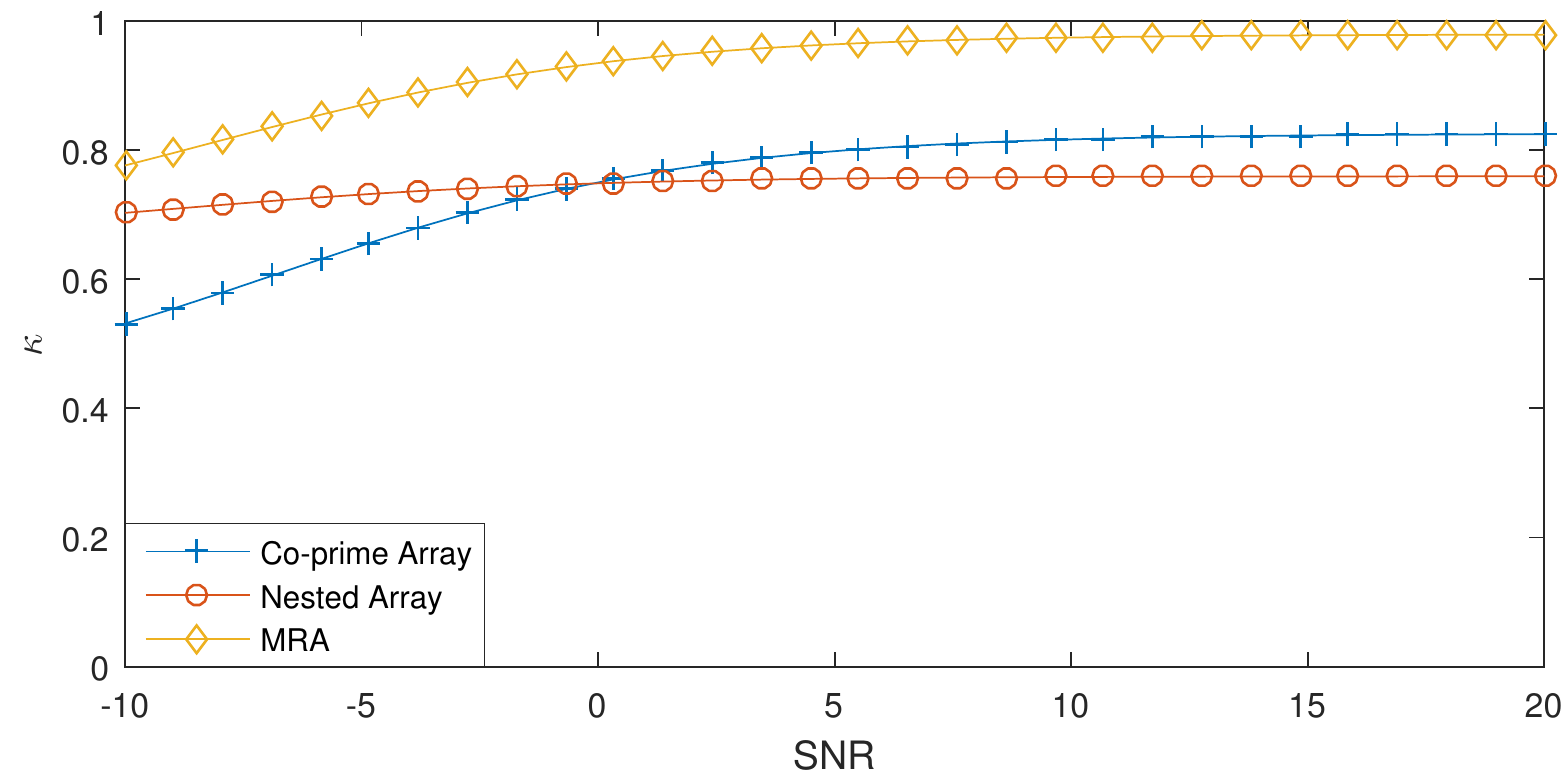}
        \label{fig:eff-1}
    }

    \subfloat[]{%
        \includegraphics[scale=0.54]{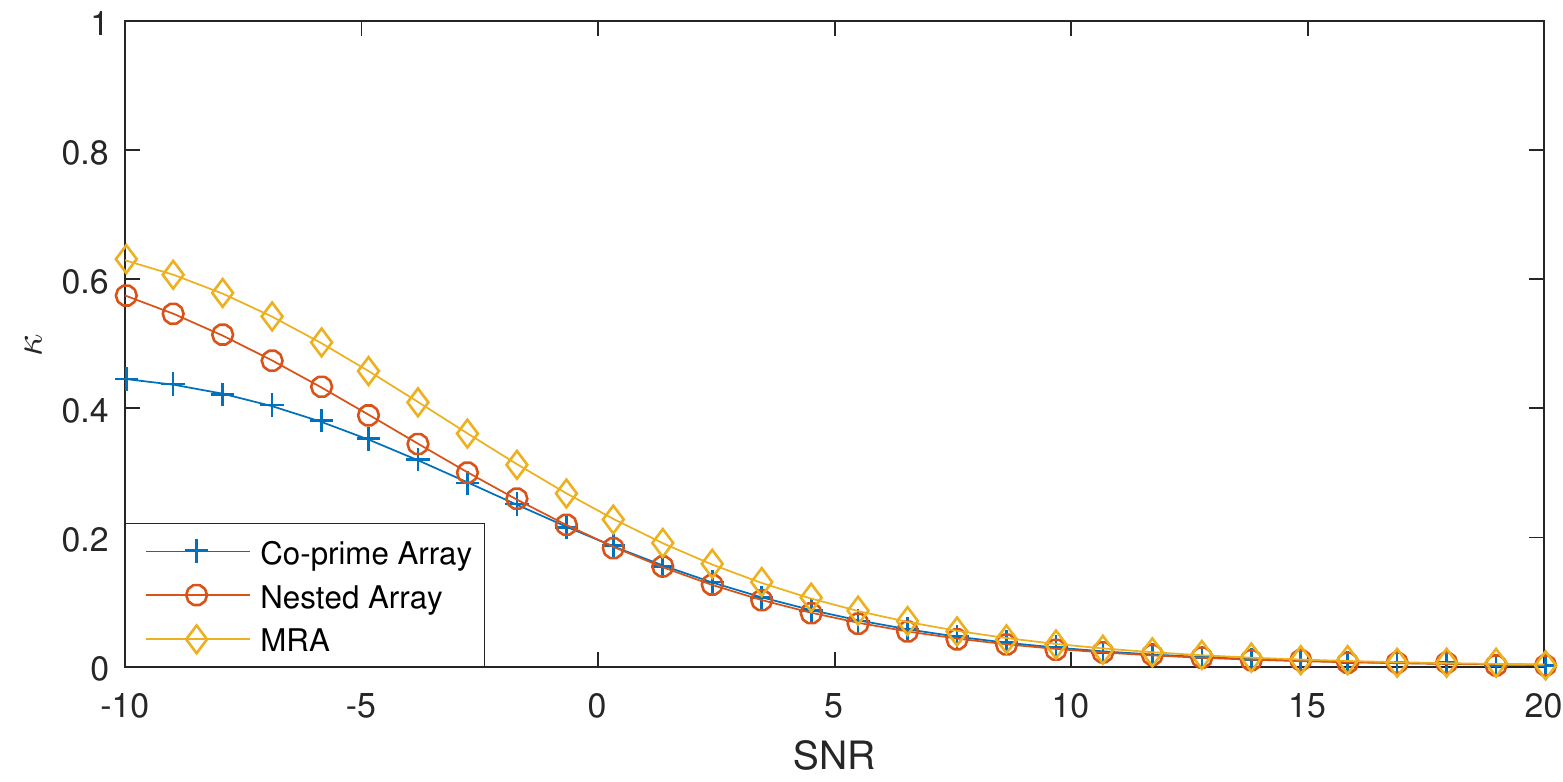}
        \label{fig:eff-6}
    }

    \subfloat[]{%
        \includegraphics[scale=0.54]{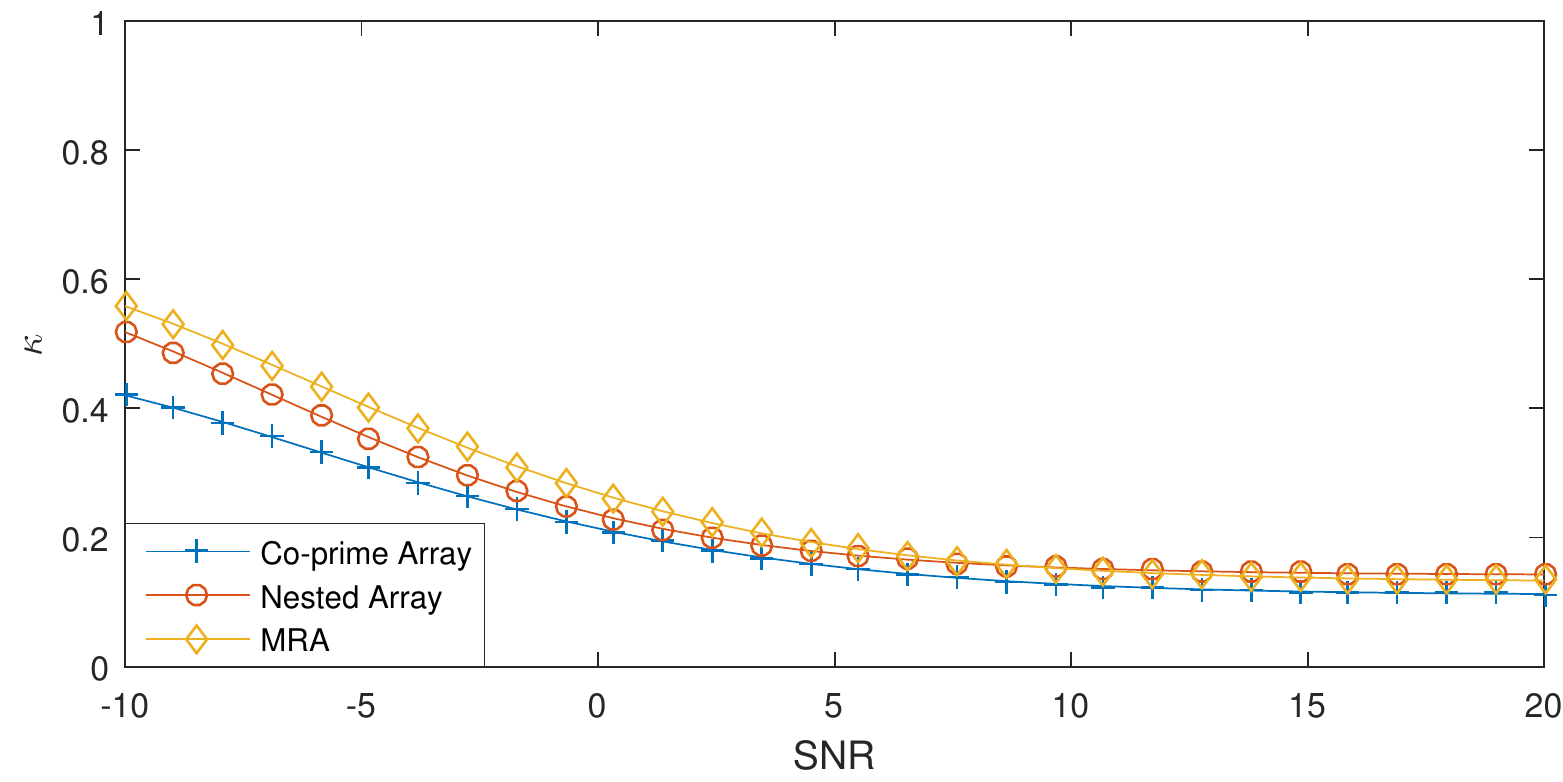}
        \label{fig:eff-12}
    }
    \caption{Average efficiency vs. SNR: (a) $K=1$, (b) $K=6$, (c) $K=12$.}    
\end{figure}

Next, we then analyze how $\kappa$ is affected by angular separation. Two sources located at $-\Delta\theta$ and $\Delta\theta$ are considered. We compute the $\kappa$ values under different choices of $\Delta\theta$ for all three arrays. For reference, we also include the empirical results obtained from 1000 trials. To satisfy the asymptotic assumption, the number of snapshots is fixed at 1000 for each trial.
As shown in Fig.~\subref*{fig:eff-sep-mra}--\subref*{fig:eff-sep-coprime}, the overall statistical efficiency decreases as the SNR increases from 0dB to 10dB for all three arrays, which is consistent with our previous observation in Fig.~\subref*{fig:eff-6}. We can also observe that the relationship between $\kappa$ and the normalized angular separation $\Delta\theta/\pi$ is rather complex, as opposed to the traditional MUSIC algorithm (c.f. \cite{stoica_music_1989}). The statistical efficiency of DA-MUSIC and SS-MUSIC is highly dependent on array geometry and angular separation.

\begin{figure}[ht]
    \centering
    \subfloat[]{%
        \includegraphics[scale=0.55]{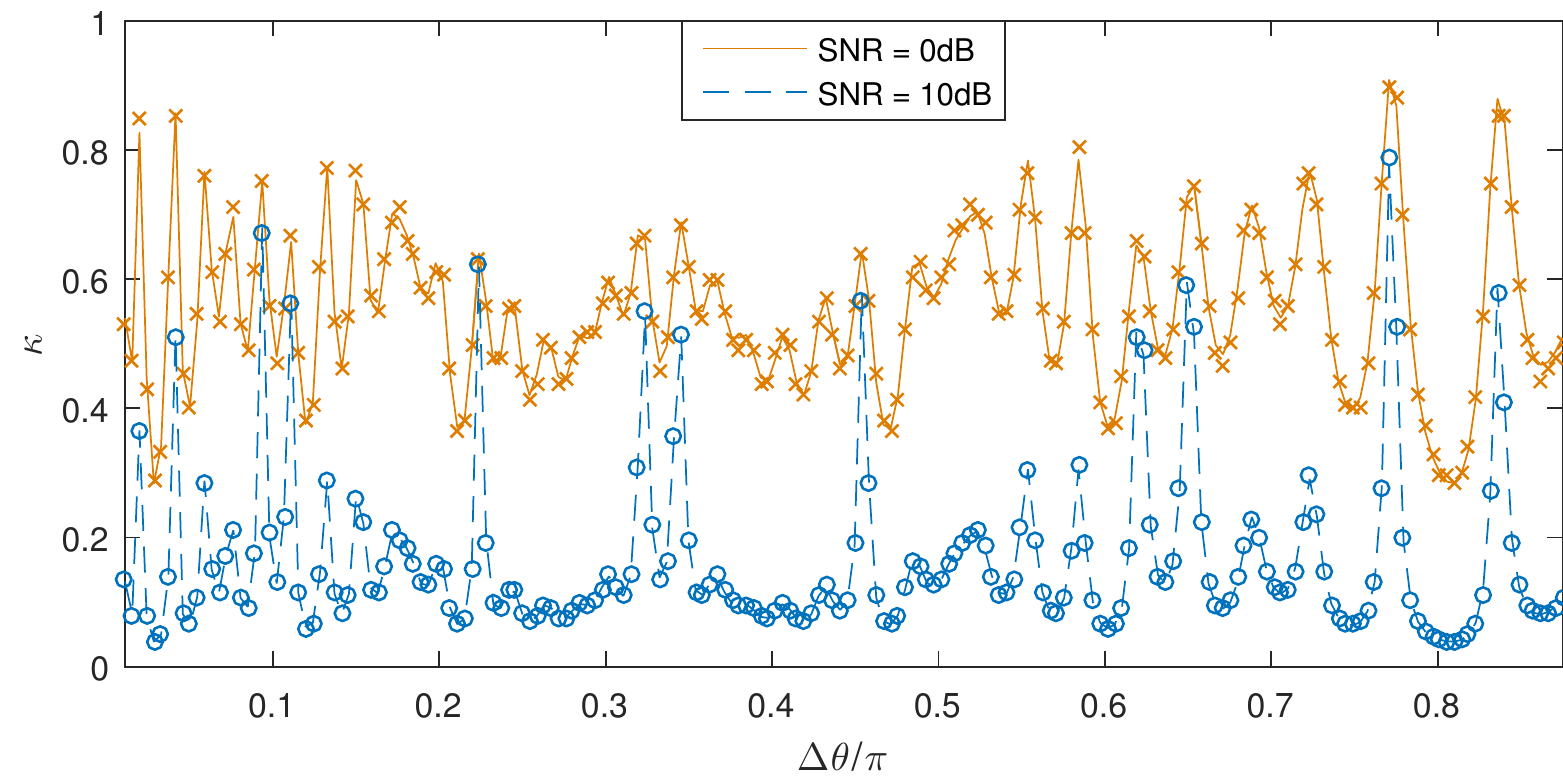}
        \label{fig:eff-sep-mra}
    }

    \subfloat[]{%
        \includegraphics[scale=0.55]{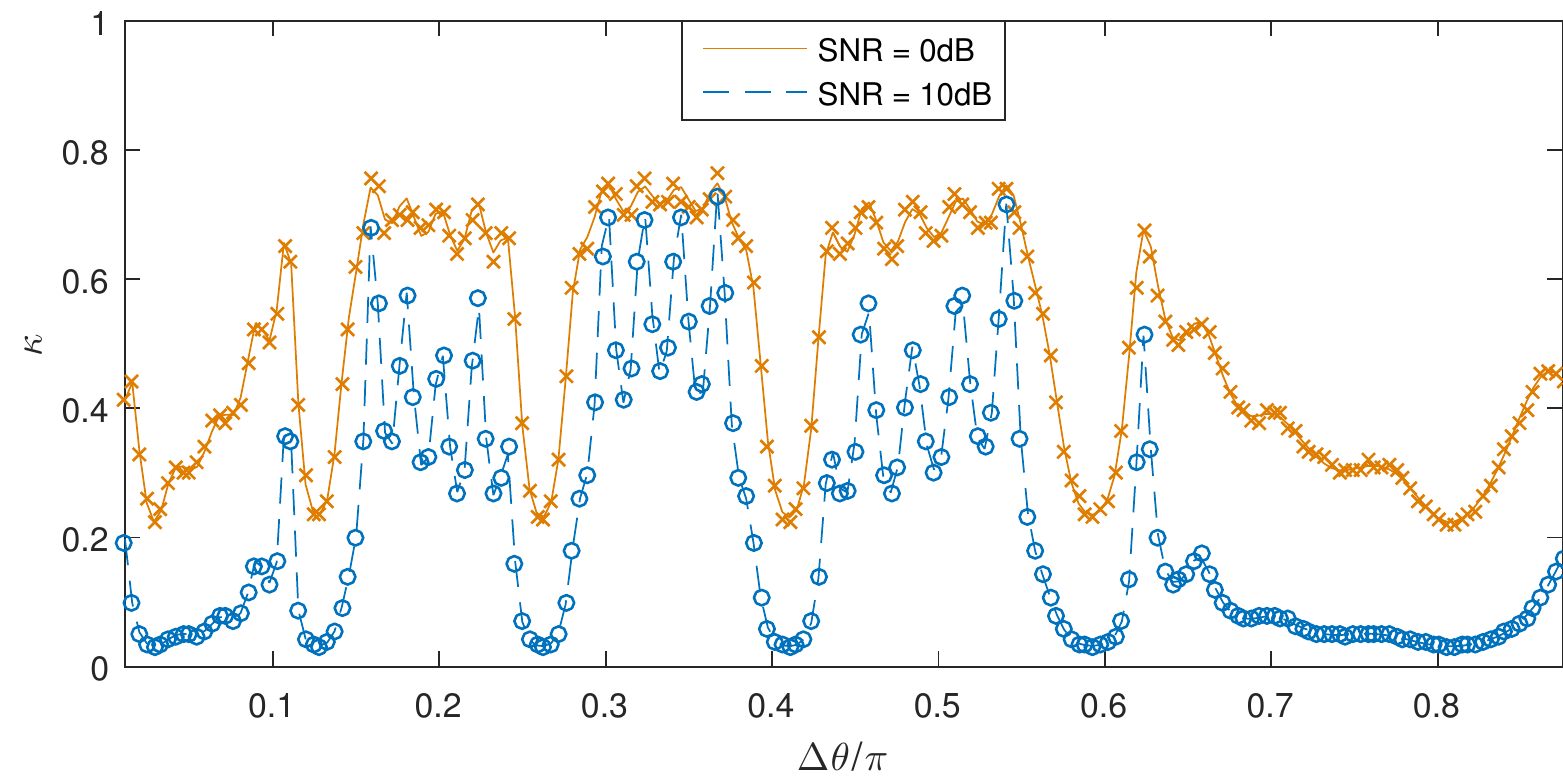}
        \label{fig:eff-sep-nested}
    }

    \subfloat[]{%
        \includegraphics[scale=0.55]{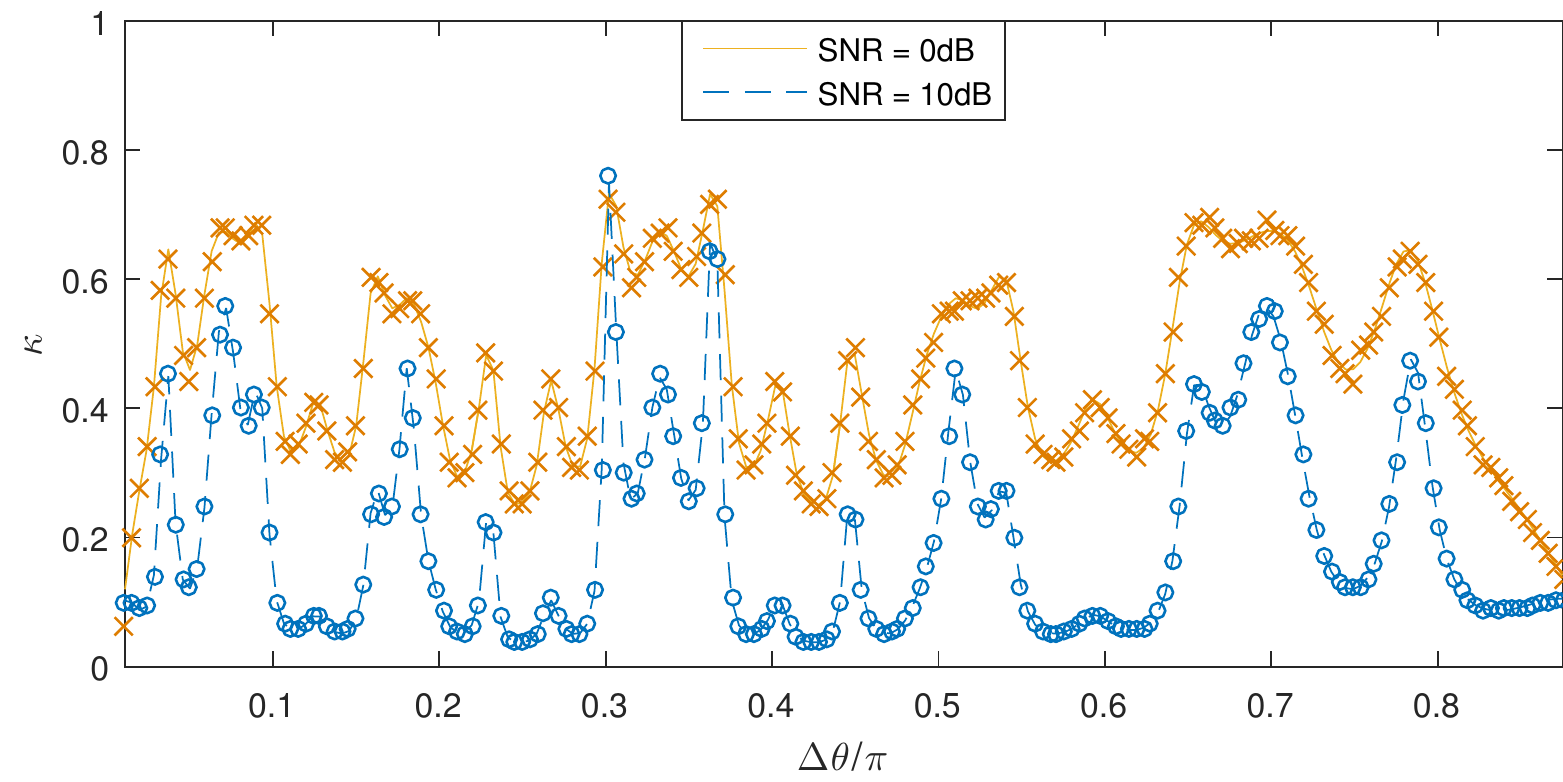}
        \label{fig:eff-sep-coprime}
    }
    \caption{Average efficiency vs. angular separation for the co-prime array: (a) MRA, (b) nested array, (c) co-prime array. The solid lines and dashed lines are analytical values obtained from \eqref{eq:avg-eff}. The circles and crosses are emprical results averaged from 1000 trials.}
\end{figure}

\subsection{MSE vs. Number of Sensors}\label{subsec:mse-n-sensor}
In this section, we investigate how the number of sensors affect the asymptotic MSE, $\epsilon(\theta_k)$. We consider three types of sparse linear arrays: co-prime arrays, nested arrays, and MRAs. In this experiment, the co-prime arrays are generated by co-prime pairs $(q, q+1)$ for $q = 2,3,\ldots, 12$. The nested arrays are generated by parameter pairs $(q+1, q)$ for $q = 2,3,\ldots,12$. The MRAs are constructed according to \cite{ishiguro_minimum_1980}. We consider two cases: the one source case where $K = 1$, and the under determined case where $K = M$. For the former case, we placed the only source at the $0^\circ$. For the later case, we placed the sources uniformly between $-60^\circ$ and $60^\circ$. We set $\SNR = 0\dB$ and $N = 1000$. The empirical MSEs were obtained from 500 trials. SS-MUSIC was used in all the trials.

\begin{figure}[h]
    \centering
    \subfloat[]{%
        \includegraphics[scale=0.56]{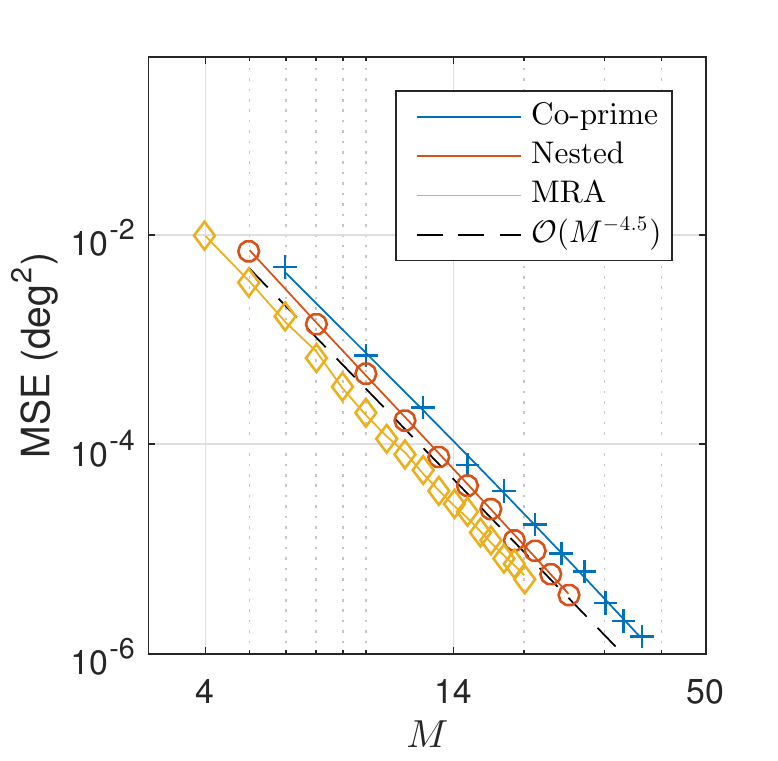}
        \label{fig:mse-n-sensor-k1}
    }
    \subfloat[]{%
        \includegraphics[scale=0.56]{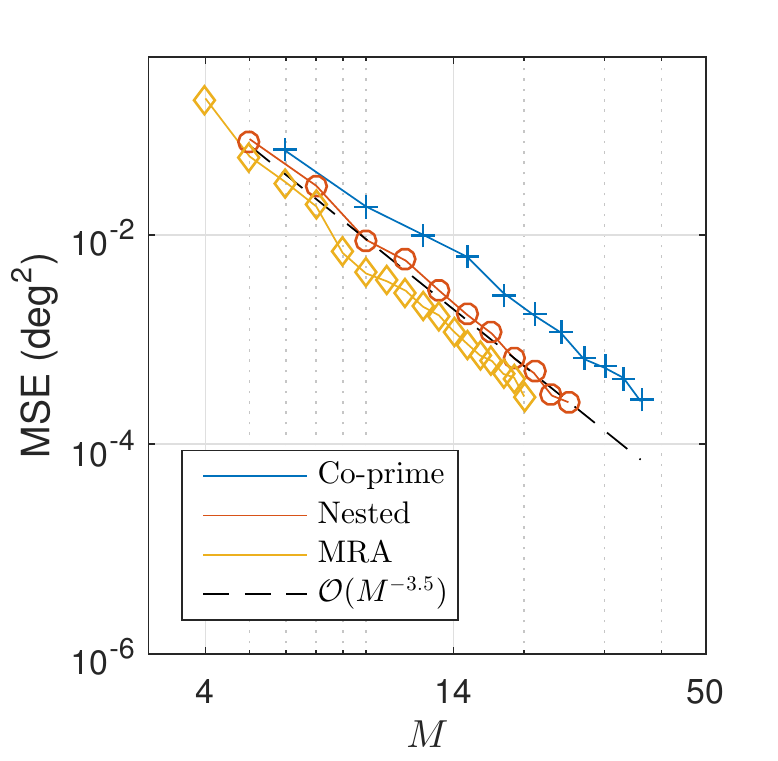}
        \label{fig:mse-n-sensor-km}
    }
    \caption{MSE vs. $M$: (a) $K=1$, (b) $K=M$. The solid lines are analytical results. The ``$+$'', ``$\circ$'', and ``$\diamond$'' denote empirical results obtains from 500 trials. The dashed lines are trend lines used for comparison.}
\end{figure}

In Fig.~\subref*{fig:mse-n-sensor-k1}, we observe that when $K = 1$, the MSE decreases at a rate of approximately $\scriptO(M^{-4.5})$ for all three arrays. In Fig.~\subref*{fig:mse-n-sensor-km}, we observe that when $K = M$, the MSE only decreases at a rate of approximately $\scriptO(M^{-3.5})$. In both cases, the MRAs and the nested arrays achieve lower MSE than the co-prime arrays. Another interesting observation is that for all three arrays, the MSE decreases faster than $\scriptO(M^{-3})$. Recall that for a $M$-sensor ULA, the asymptotic MSE of traditional MUSIC decreases at a rate of $\scriptO(M^{-3})$ as $M \to \infty$ \cite{stoica_music_1989}. This observation suggests that given the same number of sensors, these sparse linear arrays can achieve higher estimation accuracy than ULAs when the number of sensors is large.

\section{Conclusion}
In this paper, we reviewed the coarray signal model and derived the asymptotic MSE expression for two coarray-based MUSIC algorithms, namely DA-MUSIC and SS-MUSIC. We theoretically proved that the two MUSIC algorithms share the same asymptotic MSE error expression. Our analytical MSE expression is more revealing and can be applied to various types of sparse linear arrays, such as co-prime arrays, nested arrays, and MRAs. In addition, our MSE expression is also valid when the number of sources exceeds the number of sensors. We also derived the CRB for sparse linear arrays, and analyzed the statistically efficiency of typical sparse linear arrays. Our results will benefit to future research on performance analysis and optimal design of sparse linear arrays. Throughout our derivations, we assume the array is perfectly calibrated. In the future, it will be interesting to extend the results in this paper to cases when model errors are present.
Additionally, we will further investigate how the number of sensors affect the MSE and the CRB for sparse linear arrays, as well as the possibility of deriving closed form expressions in the case of large number of sensors.

\appendices
\section{Definition and Properties of the coarray selection matrix}
\label{app:f-def}
According to \eqref{eq:cov-baisc},
\begin{equation*}
    R_{mn} = \sum_{k=1}^K p_k \exp[
        j(\bar{d}_m - \bar{d}_n)\phi_k
    ] + \delta_{mn}\noisevar,
\end{equation*}
where $\delta_{mn}$ denotes Kronecker's delta. This equation implies that the $(m,n)$-th element of $\boldR$ is associated with the difference $(\bar{d}_m - \bar{d}_n)$. To capture this property, we introduce the difference matrix $\boldDelta$ such that $\Delta_{mn} = \bar{d}_m - \bar{d}_n$. We also define the weight function $\omega(n): \doubleZ \mapsto \doubleZ$ as (see \cite{pal_nested_2010} for details)
\begin{equation*}
    \omega(l) = |\{(m,n)|\Delta_{mn} = l\} |,
\end{equation*}
where $|\scriptA|$ denotes the cardinality of the set $\scriptA$. Intuitively, $\omega(l)$ counts the number of all possible pairs of $(\bar{d}_m, \bar{d}_n)$ such that $\bar{d}_m - \bar{d}_n = l$. Clearly, $\omega(l) = \omega(-l)$.

\begin{definition}
    \label{def:F}
    The coarray selection matrix $\boldF$ is a $(2\Mv-1) \times M^2$ matrix satisfying
    \begin{equation}
        \label{eq:f-def}
        F_{m,p + (q-1)M} = \begin{cases}
            \frac{1}{\omega(m - \Mv)} &, 
            \Delta_{pq} = m - \Mv,
            \\
            0 &, \mathrm{otherwise},
        \end{cases}
    \end{equation}
    for $m = 1,2,\ldots,2\Mv-1, p = 1,2,\ldots,M, q=1,2,\ldots,M$.
\end{definition}

To better illustrate the construction of $\boldF$, we consider a toy array whose sensor locations are given by $\{0, d_0, 4d_0\}$. The corresponding difference matrix of this array is
\begin{equation*}
    \boldDelta = \begin{bmatrix}
        0  & -1 & -4 \\
        1  &  0 & -3 \\
        4  &  3 &  0
    \end{bmatrix}.
\end{equation*}
The ULA part of the difference coarray consists of three sensors located at $-d_0$, $0$, and $d_0$. The weight function satisfies $\omega(-1) = \omega(1) = 1$, and $\omega(0) = 3$, so $\Mv = 2$. We can write the coarray selection matrix as
\begin{equation*}
    \boldF = \begin{bmatrix}
        0 & 0 & 0 & 
        1 & 0 & 0 & 
        0 & 0 & 0  \\
        \frac{1}{3} & 0 & 0 &
        0 & \frac{1}{3} & 0 &
        0 & 0 & \frac{1}{3}  \\
        0 & 1 & 0 &
        0 & 0 & 0 & 
        0 & 0 & 0 
    \end{bmatrix}.
\end{equation*}
If we pre-multiply the vectorized sample covariance matrix $\boldr$ by $\boldF$, we obtain the observation vector of the virtual ULA (defined in \eqref{eq:coarray-ula-model}): 
\begin{equation*}
    \boldz = \begin{bmatrix}
        z_1 \\
        z_2 \\
        z_3 \\
    \end{bmatrix} 
    =
    \begin{bmatrix}
        R_{12} \\
        \frac{1}{3}(R_{11} + R_{22} + R_{33}) \\
        R_{21}
    \end{bmatrix}.
\end{equation*}
It can be seen that $z_m$ is obtained by averaging all the elements in $\boldR$ that correspond to the difference $m - \Mv$, for $m = 1,2,\ldots,2\Mv-1$.

Based on Definition~\ref{def:F}, we now derive several useful properties of $\boldF$.
\begin{lemma}
    \label{lem:f-special-sym}
    $F_{m,p+(q-1)M} = F_{2\Mv-m, q+(p-1)M}$ for $m = 1,2,\ldots,2\Mv-1, p = 1,2,\ldots,M, q=1,2,\ldots,M$.
\end{lemma}
\begin{proof}
If $F_{m,p+(q-1)M} = 0$, then $\Delta_{pq} \neq m - \Mv$. Because $\Delta_{qp} = -\Delta_{pq}$, $\Delta_{qp} \neq -(m - \Mv)$. Hence $(2\Mv-m)-\Mv = -(m - \Mv) \neq \Delta_{qp}$, which implies that $F_{2\Mv-m, q+(p-1)M}$ is also zero.

If $F_{m,p+(q-1)M} \neq 0$, then $\Delta_{pq} = m - \Mv$ and $F_{m,p+(q-1)M} = 1/\omega(m - \Mv)$. Note that $(2\Mv-m)-\Mv = -(m - \Mv) = -\Delta_{pq} = \Delta_{qp}$. We thus have $F_{2\Mv-m, q+(p-1)M} = 1/\omega(-(m - \Mv)) = 1/\omega(m - \Mv) = F_{m,p+(q-1)M}$.
\end{proof}

\begin{lemma}
    \label{lem:fr-conj-sym}
    Let $\boldR \in \doubleC^M$ be Hermitian symmetric. Then $\boldz = \boldF \vecm(\boldR)$ is conjugate symmetric.
\end{lemma}
\begin{proof}
By Lemma~\ref{lem:f-special-sym} and $\boldR = \boldR^H$,
\begin{equation*}
    \begin{split}
        z_m &= \sum_{p=1}^{M} \sum_{q=1}^{M} F_{m, p+(q-1)M} R_{pq} \\
            &= \sum_{q=1}^{M} \sum_{p=1}^{M} F_{2\Mv - m, q+(p-1)M} R_{qp}^* \\
            &= z_{2\Mv - m}^*.
    \end{split}
\end{equation*}
\end{proof}

\begin{lemma}
    \label{lem:ftz-hermitian}
    Let $\boldz \in \doubleC^{2\Mv-1}$ be conjugate symmetric. Then $\matm_{M,M}(\boldF^T \boldz)$ is Hermitian symmetric.
\end{lemma}
\begin{proof}
Let $\boldH = \matm_{M,M}(\boldF^T \boldz)$. Then
\begin{equation}
    H_{pq} = \sum_{m=1}^{2\Mv-1} z_m F_{m,p+(q-1)M}.
\end{equation}
We know that $\boldz$ is conjugate symmetric, so $z_{m} = z_{2\Mv-m}^*$. Therefore, by Lemma~\ref{lem:f-special-sym}
\begin{equation}
    \begin{split}
            H_{pq} &= \sum_{m=1}^{2\Mv-1} z_{2\Mv-m}^* F_{2\Mv-m,q+(p-1)M} \\
            &= \Bigg[\sum_{m'=1}^{2\Mv-1} z_{m'} F_{m', q+(p-1)M}\Bigg]^* \\
            &= H_{qp}^*.
    \end{split}
\end{equation}
\end{proof}

\section{Proof of Theorem~\ref{thm:same-doa-err}}
\label{app:thm-err-expression}
We first derive the first-order expression of DA-MUSIC. Denote the eigendecomposition of $\Rvone$ by
\begin{equation*}
    \Rvone = \Es \boldLambda_\mathrm{s1} \Es^H +
        \En \boldLambda_\mathrm{n1} \En^H,
\end{equation*}
where $\En$ and $\Es$ are eigenvectors of the signal subspace and noise subspace, respectively, and $\Lsone, \Lnone$ are the corresponding eigenvalues. Specifically, we have $\boldLambda_\mathrm{n1} = \noisevar\boldI$.

Let $\Rvonet = \Rvone + \dRvone$, $\Enonet = \En + \dEnone$, and $\Lnonet = \Lnone + \dLnone$ be the perturbed versions of $\Rvone$, $\En$, and $\Lnone$. The following equality holds:
\begin{equation*}
    (\Rvone + \dRvone)(\En + \dEnone) = (\En + \dEnone)(\Lnone + \dLnone).
\end{equation*}
If the perturbation is small, we can omit high-order terms and obtain~\cite{swindlehurst_performance_1992,li_performance_1993,stewart_error_1973}
\begin{equation}
    \label{eq:AvH-delta-En-1}
    \Av^H \dEnone \doteq -\boldP^{-1} \Av^\dagger \dRvone \En.
\end{equation}
Because $\boldP$ is diagonal, for a specific $\theta_k$, we have
\begin{equation}
    \label{eq:ah-en-expression}
    \bolda^H(\theta_k) \dEnone \doteq
    -p_k^{-1} \bolde_k^T \Av^\dagger \dRvone \En,
\end{equation}
where $\bolde_k$ is the $k$-th column of the identity matrix $\boldI_{K\times K}$. Based on the conclusion in Appendix B of \cite{stoica_music_1989}, under sufficiently small perturbations, the error expression of DA-MUSIC for the $k$-th DOA is given by
\begin{equation}
    \begin{aligned}
        \label{eq:doa-err-expression-stoica}
        \hat{\theta}_k^{(1)} - \theta_k
        \doteq -\frac{\Real[\avkH \dEnone \En^H \Davk)]}{\DavkH \En \En^H \Davk},
    \end{aligned}
\end{equation}
where $\Davk = \partial\avk / \partial\theta_k$.

Substituting \eqref{eq:ah-en-expression} into \eqref{eq:doa-err-expression-stoica} gives
\begin{equation}
    \label{eq:doa-err-expression}
    \hat{\theta}_k^{(1)} - \theta_k
    \doteq -\frac{
        \Real[\bolde_k^T \Av^\dagger \dRvone \En \En^H \Davk]
    }{
        p_k \DavkH \En \En^H \Davk
    }.
\end{equation}
Because $\vecm(\boldA\boldX\boldB) = (\boldB^T \otimes \boldA) \vecm(\boldX)$ and $\En \En^H = \projp{\Av}$, we can use the notations introduced in \eqref{eq:alpha-k-def}--\eqref{eq:gamma-k-def} to express \eqref{eq:doa-err-expression} as
\begin{equation}
    \label{eq:doa-err-simple-1}
    \hat{\theta}_k^{(1)} - \theta_k
    \doteq
    -(\gamma_k p_k)^{-1} \Real[(\boldbeta_k \otimes \boldalpha_k)^T \drvone],
\end{equation}
where $\drvone = \vecm(\dRvone)$.

Note that $\Rvonet$ is constructed from $\tilde{\boldR}$. It follows that $\dRvone$ actually depends on $\Delta\boldR$, which is the perturbation part of the covariance matrix $\boldR$. By the definition of $\Rvone$, 
\begin{equation*}
    \drvone = \vecm(\begin{bmatrix}
        \boldGamma_{\Mv}\Delta\boldz &
        \cdots &
        \boldGamma_2\Delta\boldz &
        \boldGamma_1\Delta\boldz 
    \end{bmatrix})
    = \boldGamma\boldF\Delta\boldr,
\end{equation*}
where $\boldGamma = [\boldGamma_{\Mv}^T\,\boldGamma_{\Mv-1}^T\,\cdots\boldGamma_1^T]^T$ and $\Delta\boldr = \vecm(\Delta\boldR)$.

Let $\boldxi_k = \boldF^T \boldGamma^T (\boldbeta_k \otimes \boldalpha_k)$. We can now express \eqref{eq:doa-err-simple-1} in terms of $\Delta\boldr$ as
\begin{equation}
    \label{eq:doa-err-simple-2}
    \hat{\theta}_k^{(1)} - \theta_k
    \doteq -(\gamma_k p_k)^{-1} \Real(\boldxi_k^T \Delta \boldr),
\end{equation}
which completes the first part of the proof.

We next consider the first-order error expression of SS-MUSIC. From \eqref{eq:rv1-rv2-relation} we know that $\Rvtwo$ shares the same eigenvectors as $\Rvone$. Hence the eigendecomposition of $\Rvtwo$ can be expressed by
\begin{equation*}
    \Rvtwo = \Es \boldLambda_\mathrm{s2} \Es^H +
        \En \boldLambda_\mathrm{n2} \En^H,
\end{equation*}
where $\Lstwo$ and $\Lntwo$ are the eigenvalues of the signal subspace and noise subspace. Specifically, we have $\boldLambda_\mathrm{n2} = \noisevarsq/\Mv \boldI$.
Note that $\Rvtwo = (\Av\boldP\AvH + \noisevarsq\boldI)^2/\Mv$. Following a similar approach to the one we used to obtain \eqref{eq:AvH-delta-En-1}, we get
\begin{equation*}
    \AvH \dEntwo \doteq -\Mv \boldP^{-1}
    (\boldP \AvH \Av + 2\noisevar\boldI)^{-1} \Av^\dagger \dRvtwo \En,
\end{equation*}
where $\dEntwo$ is the perturbation of the noise eigenvectors produced by $\dRvtwo$. After omitting high-order terms, $\dRvtwo$ is given by
\begin{equation*}
    \dRvtwo \doteq
    \frac{1}{\Mv} \sum_{k=1}^{\Mv}
    (\boldz_k \Delta\boldz_k^H + \Delta\boldz_k \boldz_k^H).
\end{equation*}
According to~\cite{pal_nested_2010}, each subarray observation vector $\boldz_k$ can be expressed by
\begin{equation}
    \boldz_k = \Av \boldPsi^{\Mv-k} \boldp + \noisevar \boldi_{\Mv-k+1},
\end{equation}
for $k = 1,2,\ldots,\Mv$, where $\boldi_l$ is a vector of length $\Mv$ whose elements are zero except for the $l$-th element being one, and
\begin{equation*}
    \boldPsi = \diagm(e^{-j\phi_1}, e^{-j\phi_2}, \ldots, e^{-j\phi_K}).
\end{equation*}
Observe that
\begin{equation*}
    \sum_{k=1}^{\Mv} \noisevar \boldi_{\Mv-k+1} \Delta\boldz_k^H
    = \noisevar \dRvone^H,
\end{equation*}
and
\begin{equation*}
    \begin{split}
        &\sum_{k=1}^{\Mv} \Av \boldPsi^{\Mv-k}\boldp \Delta\boldz_k^H \\
        =&\Av \boldP \begin{bmatrix}
            e^{-j(\Mv-1)\phi_1} & e^{-j(\Mv-2)\phi_1} & \cdots & 1 \\
            e^{-j(\Mv-1)\phi_2} & e^{-j(\Mv-2)\phi_2} & \cdots & 1 \\
            \vdots              & \vdots              & \ddots & \vdots \\
            e^{-j(\Mv-1)\phi_K} & e^{-j(\Mv-2)\phi_K} & \cdots & 1 
        \end{bmatrix}
        \begin{bmatrix}
            \Delta\boldz_1^H \\
            \Delta\boldz_2^H \\
            \vdots \\
            \Delta\boldz_{\Mv}^H
        \end{bmatrix} \\
        =&\Av \boldP (\boldT_{\Mv} \Av)^H \boldT_{\Mv} \dRvoneH \\
        =&\Av \boldP \AvH \dRvoneH,
    \end{split}
\end{equation*}
where $\boldT_{\Mv}$ is a $\Mv \times \Mv$ permutation matrix whose anti-diagonal elements are one, and whose remaining elements are zero. Because $\Delta\boldR = \Delta\boldR^H$, by Lemma~\ref{lem:fr-conj-sym} we know that $\Delta\boldz$ is conjugate symmetric. According to the definition of $\Rvone$, it is straightforward to show that $\dRvone = \dRvone^H$ also holds. Hence
\begin{equation*}
    \dRvtwo \doteq \frac{1}{\Mv}
    [(\Av\boldP\AvH + 2\noisevar\boldI)\dRvone + 
        \dRvone\Av\boldP\AvH].
\end{equation*}
Substituting $\dRvtwo$ into the expression of $\AvH \dEntwo$, and utilizing the property that $\AvH\En = \boldzero$,
we can express $\AvH \dEntwo$ as
\begin{equation*}
    -\boldP^{-1}
    (\boldP \AvH \Av + 2\noisevar\boldI)^{-1}
    \Av^\dagger
    (\Av\boldP\AvH + 2\noisevar\boldI)\dRvone \En.
\end{equation*}
Observe that
\begin{equation*}
    \begin{aligned}
        \Av^\dagger(\Av\boldP\AvH + 2\noisevar\boldI) 
        =&(\AvH\Av)^{-1}\AvH(\Av\boldP\AvH + 2\noisevar\boldI) \\
        =&[\boldP\AvH + 2\noisevar(\AvH\Av)^{-1}\AvH] \\
        =&(\boldP\AvH\Av + 2\noisevar\boldI)\Av^\dagger.
    \end{aligned}
\end{equation*}
Hence the term $(\boldP\AvH\Av + 2\noisevar\boldI)$ gets canceled and we obtain
\begin{equation}
    \AvH \dEntwo \doteq -\boldP^{-1} \Av^\dagger \dRvone \En,
\end{equation}
which coincides with the first-order error expression of $\AvH \dEnone$.

\section{Proof of Theorem~\ref{thm:MSE-MUSIC}}
\label{app:thm-mse-music}
Before proceeding to the main proof, we introduce the following definition.
\begin{definition}
    \label{def:cab}
    Let $\boldA = [\bolda_1\,\bolda_2\,\ldots\bolda_N] \in \doubleR^{N \times N}$, and $\boldB = [\boldb_1\,\boldb_2\,\ldots\boldb_N] \in \doubleR^{N \times N}$. The structured matrix $\boldC_{\boldA\boldB} \in \doubleR^{N^2 \times N^2}$ is defined as
    \begin{equation*}
        \boldC_{\boldA\boldB} =
        \begin{bmatrix}
            \bolda_1 \boldb_1^T &
            \bolda_2 \boldb_1^T &
            \ldots &
            \bolda_N \boldb_1^T \\
            \bolda_1 \boldb_2^T &
            \bolda_2 \boldb_2^T &
            \ldots &
            \bolda_N \boldb_2^T \\
            \vdots & \ddots & \vdots & \vdots \\
            \bolda_1 \boldb_N^T &
            \bolda_2 \boldb_N^T &
            \ldots &
            \bolda_N \boldb_N^T \\
        \end{bmatrix}.
    \end{equation*}
\end{definition}

We now start deriving the explicit MSE expression. According to \eqref{eq:doa-err-simple-2},
\begin{equation}
    \label{eq:mse-expression-orig}
    \begin{aligned}
        &\doubleE[(\hat{\theta}_{k_1} - \theta_{k_1})
            (\hat{\theta}_{k_2} - \theta_{k_2})] \\
        \doteq&(\gamma_{k_1} p_{k_1})^{-1} (\gamma_{k_2} p_{k_2})^{-1}
            \doubleE[\Real(\boldxi_{k_1}^T \Delta \boldr)
            \Real(\boldxi_{k_2}^T \Delta \boldr)] \\
        =&(\gamma_{k_1} p_{k_1})^{-1} (\gamma_{k_2} p_{k_2})^{-1}
        \big\{
            \Real(\boldxi_{k_1})^T
            \doubleE[\Real(\Delta \boldr) \Real(\Delta \boldr)^T]
            \Real(\boldxi_{k_2}) \\
        &+ \Imag(\boldxi_{k_1})^T
            \doubleE[\Imag(\Delta \boldr) \Imag(\Delta \boldr)^T]
            \Imag(\boldxi_{k_2}) \\
        &- \Real(\boldxi_{k_1})^T
            \doubleE[\Real(\Delta \boldr) \Imag(\Delta \boldr)^T]
            \Imag(\boldxi_{k_2}) \\
        &- \Real(\boldxi_{k_2})^T
            \doubleE[\Real(\Delta \boldr) \Imag(\Delta \boldr)^T]
            \Imag(\boldxi_{k_1})
        \big\},
    \end{aligned}
\end{equation}
where we used the property that $\Real(\boldA\boldB) = \Real(\boldA)\Real(\boldB) - \Imag(\boldA)\Imag(\boldB)$ for two complex matrices $\boldA$ and $\boldB$ with proper dimensions.

To obtain the closed-form expression for \eqref{eq:mse-expression-orig}, we need to compute the four expectations. It should be noted that in the case of finite snapshots, $\Delta\boldr$ does not follow a circularly-symmetric complex Gaussian distribution. Therefore we cannot directly use the properties of the circularly-symmetric complex Gaussian distribution to evaluate the expectations. For brevity, we demonstrate the computation of only the first expectation in \eqref{eq:mse-expression-orig}. The computation of the remaining three expectations follows the same idea.

Let $\boldr_i$ denote the $i$-th column of $\boldR$ in \eqref{eq:cov-baisc}. Its estimate, $\hat{\boldr}_i$, is given by $\sum_{t=1}^N \boldy(t) y_i^*(t)$, where $y_i(t)$ is the $i$-th element of $\boldy(t)$. Because $\doubleE[\hat{\boldr}_i] = \boldr_i$,
\begin{equation}
    \label{eq:e-re-re-orig}
    \begin{split}
        &\doubleE[\Real(\Delta \boldr_i) \Real(\Delta \boldr_l)^T] \\
        =& \doubleE[\Real(\hat{\boldr}_i) \Real(\hat{\boldr}_l)^T]
            - \Real(\boldr_i) \Real(\boldr_l)^T.
    \end{split}
\end{equation}
The second term in \eqref{eq:e-re-re-orig} is deterministic, and the first term in \eqref{eq:e-re-re-orig} can be expanded into
\begingroup
\allowdisplaybreaks
\begin{align}
    \label{eq:e-re-re}
    & \frac{1}{N^2} \doubleE\Bigg[
        \Real\Big(\sum_{s=1}^N \boldy(s)y_i^*(s)\Big)
        \Real\Big(\sum_{t=1}^N \boldy(t)y_l^*(t)\Big)^T
        \Bigg] \nonumber \\
    =& \frac{1}{N^2} \doubleE\Bigg[
        \sum_{s=1}^N \sum_{t=1}^N
        \Real(\boldy(s)y_i^*(s))
        \Real(\boldy(t)y_l^*(t))^T
        \Bigg] \nonumber \\
    =& \frac{1}{N^2} \sum_{s=1}^N \sum_{t=1}^N \doubleE\Big\{
        \big[\Real(\boldy(s))\Real(y_i^*(s)) 
        - \Imag(\boldy(s))\Imag(y_i^*(s))\big]
         \nonumber \\
    &\quad
        \big[\Real(\boldy(t))^T\Real(y_l^*(t)) 
        - \Imag(\boldy(t))^T\Imag(y_l^*(t))\big]
        \Big\} \nonumber \\
    =& \frac{1}{N^2} \sum_{s=1}^N \sum_{t=1}^N \Big\{
        \doubleE[
            \Real(\boldy(s))\Real(y_i(s))
            \Real(\boldy(t))^T\Real(y_l(t))] \nonumber \\
    &+ \doubleE[
            \Real(\boldy(s))\Real(y_i(s))
            \Imag(\boldy(t))^T\Imag(y_l(t))] \nonumber \\
    &+ \doubleE[
            \Imag(\boldy(s))\Imag(y_i(s))
            \Real(\boldy(t))^T\Real(y_l(t))] \nonumber \\
    &+ \doubleE[
            \Imag(\boldy(s))\Imag(y_i(s))
            \Imag(\boldy(t))^T\Imag(y_l(t))]
            \Big\}.
\end{align}
\endgroup

We first consider the partial sum of the cases when $s \neq t$. By \ref{ass:a4-uc-snapshot}, $\boldy(s)$ and $\boldy(t)$ are uncorrelated Gaussians.
Recall that for $\boldx \sim \scriptC\scriptN(\boldzero, \boldSigma)$, 
\begin{equation*}
    \begin{aligned}
        \doubleE[\Real(\boldx)\Real(\boldx)^T]
        = \frac{1}{2}\Real(\boldSigma) &,\ 
        \doubleE[\Real(\boldx)\Imag(\boldx)^T]
        = -\frac{1}{2}\Imag(\boldSigma) \\
        \doubleE[\Imag(\boldx)\Real(\boldx)^T]
        = \frac{1}{2}\Imag(\boldSigma) &,\ 
        \doubleE[\Imag(\boldx)\Imag(\boldx)^T]
        = \frac{1}{2}\Real(\boldSigma).
    \end{aligned}
\end{equation*}
We have
\begin{equation*}
    \begin{split}
        &\doubleE[\Real(\boldy(s))\Real(y_i(s))
            \Real(\boldy(t))^T\Real(y_l(t))] \\
        =& \doubleE[\Real(\boldy(s))\Real(y_i(s))]
            \doubleE[\Real(\boldy(t))^T\Real(y_l(t))] \\
        =& \frac{1}{4} \Real(\boldr_i) \Real(\boldr_l)^T.
    \end{split}
\end{equation*}
Similarly, we can obtain that when $s \neq t$,
\begin{equation}
    \begin{aligned}
        \doubleE[\Real(\boldy(s))\Real(y_i(s))
            \Imag(\boldy(t))^T\Imag(y_l(t))]
        &= \frac{1}{4} \Real(\boldr_i) \Real(\boldr_l)^T, \\
        \doubleE[\Imag(\boldy(s))\Imag(y_i(s))
            \Real(\boldy(t))^T\Real(y_l(t))]
        &= \frac{1}{4} \Real(\boldr_i) \Real(\boldr_l)^T, \\
        \doubleE[\Imag(\boldy(s))\Imag(y_i(s))
            \Imag(\boldy(t))^T\Imag(y_l(t))]
        &= \frac{1}{4} \Real(\boldr_i) \Real(\boldr_l)^T. \\
    \end{aligned}
\end{equation}
Therefore the partial sum of the cases when $s \neq t$ is given by $(1-1/N) \Real(\boldr_i) \Real(\boldr_l)^T$.

We now consider the partial sum of the cases when $s = t$. We first consider the first expectation inside the double summation in \eqref{eq:e-re-re}. Recall that for $\boldx \sim \scriptN(\boldzero, \boldSigma)$, $\doubleE[x_i x_l x_p x_q] = \sigma_{il}\sigma_{pq} + \sigma_{ip}\sigma_{lq} + \sigma_{iq}\sigma_{lp}$. 
We can express the $(m,n)$-th element of the matrix $\doubleE[\Real(\boldy(t))\Real(y_i(t))\Real(\boldy(t))^T\Real(y_l(t))]$ as
\begin{align*}
    &\doubleE[\Real(y_m(t))\Real(y_i(t))
        \Real(y_n(t))\Real(y_l(t))] \nonumber \\
    =&\doubleE[\Real(y_m(t))\Real(y_i(t))
        \Real(y_l(t))\Real(y_n(t))] \nonumber \\
    =&\doubleE[\Real(y_m(t))\Real(y_i(t))]
        \doubleE[\Real(y_l(t))\Real(y_n(t))] \\
    &+ \doubleE[\Real(y_m(t))\Real(y_l(t))]
        \doubleE[\Real(y_i(t))\Real(y_n(t))] \nonumber \\
    &+ \doubleE[\Real(y_m(t))\Real(y_n(t))]
        \doubleE[\Real(y_i(t))\Real(y_l(t))] \nonumber \\
    =& \frac{1}{4}[\Real(R_{mi})\Real(R_{ln})
        + \Real(R_{ml})\Real(R_{in})
        + \Real(R_{mn})\Real(R_{il})]. \nonumber
\end{align*}
Hence
\begin{equation*}
    \begin{split}
        &\doubleE[\Real(\boldy(t))\Real(y_i(t)) 
            \Real(\boldy(t))^T\Real(y_l(t))] \\
        =& \frac{1}{4}[\Real(\boldr_i)\Real(\boldr_l)^T
            + \Real(\boldr_l)\Real(\boldr_i)^T
            + \Real(\boldR)\Real(R_{il})].
    \end{split}
\end{equation*}
Similarly, we obtain that
\begin{equation*}
    \begin{split}
        &\doubleE[\Imag(\boldy(t))\Imag(y_i(t)) 
            \Imag(\boldy(t))^T\Imag(y_l(t))] \\
        =& \frac{1}{4}[\Real(\boldr_i)\Real(\boldr_l)^T
            + \Real(\boldr_l)\Real(\boldr_i)^T
            + \Real(\boldR)\Real(R_{il})],
    \end{split}
\end{equation*}
\begin{equation*}
    \begin{split}
        &\doubleE[\Real(\boldy(t))\Real(y_i(t)) 
            \Imag(\boldy(t))^T\Imag(y_l(t))] \\
        =& \doubleE[\Imag(\boldy(t))\Imag(y_i(t)) 
            \Real(\boldy(t))^T\Real(y_l(t))] \\
        =& \frac{1}{4}[\Real(\boldr_i)\Real(\boldr_l)^T
            - \Imag(\boldr_l)\Imag(\boldr_i)^T
            + \Imag(\boldR)\Imag(R_{il})].
    \end{split}
\end{equation*}
Therefore the partial sum of the cases when $s = t$ is given by
$(1/N)\Real(\boldr_i)\Real(\boldr_l)^T + (1/2N)[\Real(\boldR)\Real(R_{il}) + \Imag(\boldR)\Imag(R_{il}) + \Real(\boldr_l)\Real(\boldr_i)^T - \Imag(\boldr_l)\Imag(\boldr_i)^T]$
. Combined with the previous partial sum of the cases when $s \neq t$, we obtain that
\begin{equation}
    \begin{split}
        &\doubleE[\Real(\Delta \boldr_i) \Real(\Delta \boldr_l)^T] \\
        =&\frac{1}{2N}[\Real(\boldR)\Real(R_{il})
            + \Imag(\boldR)\Imag(R_{il}) \\
        &+ \Real(\boldr_l)\Real(\boldr_i)^T
        - \Imag(\boldr_l)\Imag(\boldr_i)^T ].
    \end{split}
\end{equation}
Therefore
\begin{equation}
    \label{eq:e-re-re-middle}
    \begin{split}
        &\doubleE[\Real(\Delta \boldr) \Real(\Delta \boldr)^T] \\
        =& \frac{1}{2N}[\Real(\boldR) \otimes \Real(\boldR)
            + \Imag(\boldR) \otimes \Imag(\boldR) \\
            &\quad+ \boldC_{\Real(\boldR)\Real(\boldR)}
            - \boldC_{\Imag(\boldR)\Imag(\boldR)}],
    \end{split}
\end{equation}
which completes the computation of first expectation in \eqref{eq:mse-expression-orig}. Utilizing the same technique, we obtain that
\begin{equation}
    \label{eq:e-im-im-middle}
    \begin{split}
        &\doubleE[\Imag(\Delta \boldr) \Imag(\Delta \boldr)^T] \\
        =& \frac{1}{2N}[\Real(\boldR) \otimes \Real(\boldR)
            + \Imag(\boldR) \otimes \Imag(\boldR) \\
            &\quad+ \boldC_{\Imag(\boldR)\Imag(\boldR)}
            - \boldC_{\Real(\boldR)\Real(\boldR)}],
    \end{split}
\end{equation}
and
\begin{equation}
    \label{eq:e-re-im-middle}
    \begin{split}
        &\doubleE[\Real(\Delta \boldr) \Imag(\Delta \boldr)^T] \\
        =& \frac{1}{2N}[\Imag(\boldR) \otimes \Real(\boldR)
            - \Real(\boldR) \otimes \Imag(\boldR) \\
            &\quad+ \boldC_{\Real(\boldR)\Imag(\boldR)}
            + \boldC_{\Imag(\boldR)\Real(\boldR)}].
    \end{split}
\end{equation}
Substituting \eqref{eq:e-re-re-middle}--\eqref{eq:e-re-im-middle} into \eqref{eq:mse-expression-orig} gives a closed-form MSE expression. However, this expression is too complicated for analytical study. In the following steps, we make use of the properties of $\boldxi_k$ to simply the MSE expression.


\begin{lemma}
    \label{lem:aka-caa}
    Let $\boldX, \boldY, \boldA, \boldB \in \doubleR^{N \times N}$ satisfying $\boldX^T = (-1)^{n_x}\boldX$, $\boldA^T = (-1)^{n_a}\boldA$, and $\boldB^T = (-1)^{n_b}\boldB$, where $n_x, n_a, n_b \in \{0,1\}$. Then
    \begin{equation*}
        \vecm(\boldX)^T (\boldA \otimes \boldB) \vecm(\boldY)
        = (-1)^{n_x+n_b}\vecm(\boldX)^T \boldC_{\boldA\boldB} \vecm(\boldY),
    \end{equation*}
    \begin{equation*}
        \vecm(\boldX)^T (\boldB \otimes \boldA) \vecm(\boldY)
        = (-1)^{n_x+n_a}\vecm(\boldX)^T \boldC_{\boldB\boldA} \vecm(\boldY).
    \end{equation*}
\end{lemma}
\begin{proof}
    By Definition~\ref{def:cab}, 
    \begingroup
    \allowdisplaybreaks
    \begin{align*}
        &\vecm(\boldX)^T \boldC_{\boldA \boldB} \vecm(\boldY) \\
        =& \sum_{m=1}^N \sum_{n=1}^N 
            \boldx_m^T \bolda_n \boldb_m^T \boldy_n \\
        =& \sum_{m=1}^N \sum_{n=1}^N 
            \Big( \sum_{p=1}^N A_{pn} X_{pm} \Big)
            \Big( \sum_{p=1}^N B_{qm} Y_{qn} \Big) \\
        =& \sum_{m=1}^N \sum_{n=1}^N \sum_{p=1}^N \sum_{q=1}^N
            A_{pn} X_{pm} B_{qm} Y_{qn} \\
        =& (-1)^{n_x+n_b}
            \sum_{p=1}^N \sum_{n=1}^N \sum_{m=1}^N \sum_{q=1}^N
            (X_{mp} B_{mq} Y_{qn}) A_{pn} \\
        =& (-1)^{n_x+n_b}
            \sum_{p=1}^N \sum_{n=1}^N
             \boldx_p^T A_{pn} \boldB \boldy_n \\
        =& (-1)^{n_x+n_b}
            \vecm(\boldX)^T (\boldA \otimes \boldB)
            \vecm(\boldY).
    \end{align*}
    \endgroup
    The proof of the second equality follows the same idea.
\end{proof}

\begin{lemma}
    \label{lem:flip-pn}
    $\TMv \projp{\Av} \TMv = (\projp{\Av})^*$.
\end{lemma}
\begin{proof}
Since $\projp{\Av} = \boldI - \Av (\Av^H \Av)^{-1} \Av^H$, it suffices to show that $\TMv \Av (\Av^H \Av)^{-1} \Av^H \TMv = (\Av (\Av^H \Av)^{-1} \Av^H)^*$. Because $\Av$ is the steering matrix of a ULA with $\Mv$ sensors, it is straightforward to show that $\TMv \Av = (\Av \boldPhi)^*$, where $\boldPhi = \diagm(e^{-j(\Mv-1)\phi_1}, e^{-j(\Mv-1)\phi_2}, \ldots, e^{-j(\Mv-1)\phi_K})$.

Because $\TMv\TMv = \boldI, \TMv^H = \TMv$,
\begin{equation*}
    \begin{split}
        &\TMv \Av (\Av^H \Av)^{-1} \Av^H \TMv \\
        =& \TMv \Av (\Av^H \TMv^H \TMv \Av)^{-1} 
            \Av^H \TMv^H \\
        =& (\Av \boldPhi)^* ((\Av \boldPhi)^T (\Av \boldPhi)^*)^{-1} 
            (\Av \boldPhi)^T \\
        =& (\Av (\Av^H \Av)^{-1} \Av^H)^*.
    \end{split}
\end{equation*}
\end{proof}

\begin{lemma}
    \label{lem:xi-k-symmetry}
    Let $\boldXi_k = \matm_{M, M}(\boldxi_k)$. Then $\boldXi_k^H = \boldXi_k$ for $k = 1,2,\ldots, K$.
\end{lemma}
\begin{proof}
Note that $\boldxi_k = \boldF^T \boldGamma^T (\boldbeta_k \otimes \boldalpha_k)$. We first prove that $\boldbeta_k \otimes \boldalpha_k$ is conjugate symmetric, or that $(\TMv \otimes \TMv)(\boldbeta_k \otimes \boldalpha_k) = (\boldbeta_k \otimes \boldalpha_k)^*$. Similar to the proof of Lemma~\ref{lem:flip-pn}, we utilize the properties that $\TMv \Av = (\Av \boldPhi)^*$ and that $\TMv \avk = (\avk e^{-j(\Mv - 1)\phi_k})^*$ to show that
\begin{equation}
    \label{eq:flip-av}
    \TMv (\Av^\dagger)^H \bolde_k \avkH \TMv
    = [(\Av^\dagger)^H \bolde_k \avkH]^*.
\end{equation}

Observe that $\Davk = j\dot{\phi}_k \boldD \avk$, where $\dot{\phi}_k = (2\pi d_0\cos\theta_k) / \lambda$ and $\boldD = \diagm(0,1,\ldots,\Mv-1)$. We have
\begin{equation*}
    \begin{split}
        &(\TMv \otimes \TMv)(\boldbeta_k \otimes \boldalpha_k)
        = (\boldbeta_k \otimes \boldalpha_k)^* \\
        \iff& \TMv \boldalpha_k \boldbeta_k^T \TMv
        = (\boldalpha_k \boldbeta_k^T)^* \\
        \iff& \TMv[(\Av^\dagger)^H \bolde_k \avkH \boldD \projp{\Av}]^* \TMv \\
        &= -(\Av^\dagger)^H \bolde_k \avkH \boldD \projp{\Av}.
    \end{split}
\end{equation*}
Since $\boldD = \TMv \TMv \boldD \TMv \TMv$, combining with Lemma~\ref{lem:flip-pn} and \eqref{eq:flip-av}, it suffices to show that
\begin{equation}
    \label{eq:flip-akb-requirement}
    \begin{split}
        &(\Av^\dagger)^H \bolde_k \avkH \TMv \boldD
            \TMv \projp{\Av} \\
        &= - (\Av^\dagger)^H \bolde_k \avkH \boldD \projp{\Av}.
    \end{split}
\end{equation}
Observe that $\TMv \boldD \TMv + \boldD = (\Mv-1)\boldI$. We have
\begin{equation*}
    \projp{\Av}(\TMv \boldD \TMv + \boldD) \avk = \boldzero,
\end{equation*}
or equivalently
\begin{equation}
    \label{eq:flip-akb-final}
    \avkH \TMv \boldD \TMv \projp{\Av} =
    - \avkH \boldD \projp{\Av}.
\end{equation}
Pre-multiplying both sides of \eqref{eq:flip-akb-final} with $(\Av^\dagger)^H \bolde_k$ leads to \eqref{eq:flip-akb-requirement}, which completes the proof that $\boldbeta_k \otimes \boldalpha_k$ is conjugate symmetric. 
According to the definition of $\boldGamma$ in \eqref{eq:gamma-mat-def}, it is straightforward to show that $\boldGamma^T (\boldbeta_k \otimes \boldalpha_k)$ is also conjugate symmetric. Combined with Lemma~\ref{lem:ftz-hermitian} in Appendix~\ref{app:f-def}, we conclude that $\matm_{M,M}(\boldF^T \boldGamma^T (\boldbeta_k \otimes \boldalpha_k))$ is Hermitian symmetric, or that $\boldXi_k = \boldXi_k^H$.
\end{proof}

Given Lemma~\ref{lem:aka-caa}--\ref{lem:xi-k-symmetry}, we are able continue the simplification. We first consider the term $\Real(\boldxi_{k_1})^T \doubleE[\Real(\Delta \boldr) \Real(\Delta \boldr)^T] \Real(\boldxi_{k_2})$ in \eqref{eq:mse-expression-orig}. Let $\boldXi_{k_1} = \matm_{M,M}(\boldxi_{k_1})$, and $\boldXi_{k_2} = \matm_{M,M}(\boldxi_{k_2})$. By Lemma~\ref{lem:xi-k-symmetry}, we have $\boldXi_{k_1} = \boldXi_{k_1}^H$, and $\boldXi_{k_2} = \boldXi_{k_2}^H$. Observe that $\Real(\boldR)^T = \Real(\boldR)$, and that $\Imag(\boldR)^T = \Imag(\boldR)$. By Lemma~\ref{lem:aka-caa} we immediately obtain the following equalities:
\begin{equation*}
    \begin{split}
        &\Real(\boldxi_{k_1})^T
        (\Real(\boldR) \otimes \Real(\boldR))
        \Real(\boldxi_{k_2}) \\
        =&
        \Real(\boldxi_{k_1})^T
        \boldC_{\Real(\boldR)\Real(\boldR)}
        \Real(\boldxi_{k_2}),
    \end{split}
\end{equation*}
\begin{equation*}
    \begin{split}
        &\Real(\boldxi_{k_1})^T
        (\Imag(\boldR) \otimes \Imag(\boldR))
        \Real(\boldxi_{k_2}) \\
        =&
        -\Real(\boldxi_{k_1})^T
        \boldC_{\Imag(\boldR)\Imag(\boldR)}
        \Real(\boldxi_{k_2}).
    \end{split}
\end{equation*}
Therefore $\Real(\boldxi_{k_1})^T \doubleE[\Real(\Delta \boldr) \Real(\Delta \boldr)^T] \Real(\boldxi_{k_2})$ can be compactly expressed as
\begin{equation}
    \label{eq:r1t-err-r2-final}
    \begin{aligned}
        &\Real(\boldxi_{k_1})^T
        \doubleE[\Real(\Delta \boldr) \Real(\Delta \boldr)^T] 
        \Real(\boldxi_{k_2}) \\
        =& \frac{1}{N}
            \Real(\boldxi_{k_1})^T[
                \Real(\boldR) \otimes \Real(\boldR)
                + \Imag(\boldR) \otimes \Imag(\boldR)
            ]\Real(\boldxi_{k_2}) \\
        =& \frac{1}{N}
            \Real(\boldxi_{k_1})^T
            \Real(\boldR^T \otimes \boldR)
            \Real(\boldxi_{k_2}),
    \end{aligned}
\end{equation}
where we make use of the properties that $\boldR^T = \boldR^*$, and $\Real(\boldR^* \otimes \boldR) = \Real(\boldR) \otimes \Real(\boldR) + \Imag(\boldR) \otimes \Imag(\boldR)$. Similarly, we can obtain that
\begin{equation}
    \label{eq:i1t-err-i2-final}
    \begin{split}
        &\Imag(\boldxi_{k_1})^T 
            \doubleE[\Imag(\Delta \boldr) \Imag(\Delta \boldr)^T]
        \Imag(\boldxi_{k_2}) \\
        =& \frac{1}{N}
            \Imag(\boldxi_{k_1})^T
            \Real(\boldR^T \otimes \boldR)
            \Imag(\boldxi_{k_2}),
    \end{split}
\end{equation}
\begin{equation}
    \label{eq:r1t-eri-i2-final}
    \begin{split}
        &\Real(\boldxi_{k_1})^T 
            \doubleE[\Real(\Delta \boldr) \Imag(\Delta \boldr)^T]
        \Imag(\boldxi_{k_2}) \\
        =& -\frac{1}{N}
            \Real(\boldxi_{k_1})^T
            \Imag(\boldR^T \otimes \boldR)
            \Imag(\boldxi_{k_2}),
    \end{split}
\end{equation}
\begin{equation}
    \label{eq:r2t-eri-i1-final}
    \begin{split}
        &\Real(\boldxi_{k_2})^T 
            \doubleE[\Real(\Delta \boldr) \Imag(\Delta \boldr)^T]
        \Imag(\boldxi_{k_1}) \\
        =& -\frac{1}{N}
            \Real(\boldxi_{k_2})^T
            \Imag(\boldR^T \otimes \boldR)
            \Imag(\boldxi_{k_1}).
    \end{split}
\end{equation}
Substituting \eqref{eq:r1t-err-r2-final}--\eqref{eq:r2t-eri-i1-final} into \eqref{eq:mse-expression-orig} completes the proof.

\section{Proof of Proposition~\ref{prop:crb-snr-infty}}
\label{app:crb-snr-infty}
Without loss of generality, let $p = 1$ and $\noisevar \to 0$. For brevity, we denote $\boldR^T \otimes \boldR$ by $\boldW$. We first consider the case when $K < M$.
Denote the eigendecomposition of $\boldR^{-1}$ by $\Es \Ls^{-1} \Es^H + \noisevarinv \En \En^H$. We have
\begin{equation*}
    \boldW^{-1} = \noisevarp{-4} \boldK_1 + \noisevarp{-2} \boldK_2 + \boldK_3,
\end{equation*}
where
\begin{align*}
    \boldK_1 &= \En^*\En^T \otimes \En\En^H, \\
    \boldK_2 &= \Es^*\Ls^{-1}\Es^T \otimes \En\En^H + \En^*\En^T \otimes \Es\Ls^{-1}\Es^H, \\
    \boldK_3 &= \Es^*\Ls^{-1}\Es^T \otimes \Es\Ls^{-1}\Es^H.
\end{align*}
Recall that $\boldA^H \En = \boldzero$. We have
\begingroup
\allowdisplaybreaks
\begin{align}
    \boldK_1 \DAd &= (\En^*\En^T \otimes \En\En^H)
        (\dot{\boldA}^* \odot \boldA + \boldA^* \odot \dot{\boldA}) \nonumber\\
    &= \En^*\En^T \dot{\boldA}^* \odot \En\En^H \boldA
     + \En^*\En^T \boldA^* \odot \En\En^H \dot{\boldA} \nonumber\\
    &= \boldzero.
\end{align}
\endgroup
Therefore
\begin{equation}
    \boldM_{\boldtheta}^H \boldM_{\boldtheta} 
    = \DAd^H \boldW^{-1} \DAd 
    = \noisevarinv \DAd^H (\boldK_2 + \noisevar \boldK_3) \DAd.
\end{equation}
Similar to $\boldW^{-1}$, we denote $\boldW^{-\frac{1}{2}} = \noisevarp{-2} \boldK_1 + \noisevarp{-1} \boldK_4 + \boldK_5$, where
\begin{align*}
    \boldK_4 &= \Es^*\Ls^{-\frac{1}{2}}\Es^T \otimes 
        \En\En^H + \En^*\En^T \otimes \Es\Ls^{-\frac{1}{2}}\Es^H, \\
    \boldK_5 &= \Es^*\Ls^{-\frac{1}{2}}\Es^T \otimes \Es\Ls^{-\frac{1}{2}}\Es^H.
\end{align*}
Therefore
\begin{align*}
    &\boldM_{\boldtheta}^H \proj{\boldM_{\bolds}} \boldM_{\boldtheta} \\
    =& \DAd^H \boldW^{-\frac{1}{2}} \proj{\boldM_{\bolds}} \boldW^{-\frac{1}{2}} \DAd \\
    =& \noisevarp{-2} \DAd^H (\noisevarsqrt \boldK_5 + \boldK_4) \proj{\boldM_{\bolds}}
        (\noisevarsqrt \boldK_5 + \boldK_4) \DAd,
\end{align*}
where $\proj{\boldM_{\bolds}} = \boldM_{\bolds} \boldM_{\bolds}^\dagger$. We can then express the CRB as
\begin{equation}
    \label{eq:crb-q1-q2-q3}
    \CRB_{\boldtheta}
    = \noisevar (\boldQ_1 + \noisevarsqrt \boldQ_2 + \noisevar \boldQ_3)^{-1},
\end{equation}
where
\begin{align*}
    \boldQ_1 &= \DAd^H (\boldK_2 - \boldK_4 \proj{\boldM_{\bolds}} \boldK_4) \DAd, \\
    \boldQ_2 &= - \DAd^H (\boldK_4 \proj{\boldM_{\bolds}} \boldK_5 +
                    \boldK_5 \proj{\boldM_{\bolds}} \boldK_4) \DAd, \\
    \boldQ_3 &= \DAd^H (\boldK_3 - \boldK_5 \proj{\boldM_{\bolds}} \boldK_5) \DAd.
\end{align*}
When $\noisevar = 0$, $\boldR$ reduces to $\boldA\boldA^H$. Observe that the eigendecomposition of $\boldR$ always exists for $\noisevar \geq 0$. We use $\boldK_1^\star$--$\boldK_5^\star$ to denote the corresponding $\boldK_1$--$\boldK_5$ when $\noisevar \to 0$.

\begin{lemma}
\label{lem:proj-ms-exist}
Let $K < M$. Assume $\partial\boldr / \partial\boldeta$ is full column rank. Then $\lim_{\noisevar \to 0^+} \proj{\boldM_{\bolds}}$ exists. 
\end{lemma}
\begin{proof}
Because $\boldA^H \En = \boldzero$,
\begin{align*}
    \boldK_2 \Ad
    =& (\Es^*\Ls^{-1}\Es^T \otimes \En\En^H)(\boldA^* \odot \boldA) \\
    &+ (\En^*\En^T \otimes \Es\Ls^{-1}\Es^H)(\boldA^* \odot \boldA) \\
    =&  \Es^*\Ls^{-1}\Es^T \boldA^* \odot \En\En^H \boldA \\
    &+ \En^*\En^T \boldA^* \odot \Es\Ls^{-1}\Es^H \boldA \\
    =& \boldzero
\end{align*}
Similarly, we can show that $\boldK_4 \Ad = \boldzero$, $\boldi^H \boldK_2 \boldi = \boldi^H \boldK_4 \boldi = 0$, and $\boldi^H \boldK_1 \boldi = \mathrm{rank}(\En) = M-K$. Hence
\begin{equation*}
    \boldM_{\bolds}^H \boldM_{\bolds}
    = \begin{bmatrix}
        \Ad^H \boldK_3 \Ad & \Ad^H \boldK_3 \boldi \\
        \boldi^H \boldK_3 \Ad  & \boldi^H \boldW^{-1} \boldi
    \end{bmatrix}.
\end{equation*}
Because $\partial\boldr / \partial\boldeta$ is full column rank, $\boldM_{\bolds}^H \boldM_{\bolds}$ is full rank and positive definite. Therefore the Schur complements exist, and we can inverse $\boldM_{\bolds}^H \boldM_{\bolds}$ block-wisely. Let $\boldV = \Ad^H \boldK_3 \Ad$ and $v = \boldi^H \boldW^{-1} \boldi$. After tedious but straightforward computation, we obtain
\begin{align*}
    \proj{\boldM_{\bolds}}
    =& \boldK_5 \Ad \boldS^{-1} \Ad^H \boldK_5 \\
    &- s^{-1} \boldK_5 \Ad \boldV^{-1} \Ad^H \boldK_3 \boldi \boldi^H 
        (\boldK_5  + \noisevarinv \boldK_1) \\
    &- v^{-1} (\boldK_5 + \noisevarinv \boldK_1)
        \boldi \boldi^H \boldK_3 \Ad \boldS^{-1} \Ad^H \boldK_5 \\
    &+ s^{-1} (\boldK_5 + \noisevarinv \boldK_1) \boldi \boldi^H
         (\boldK_5  + \noisevarinv \boldK_1),
\end{align*}
where $\boldS$ and $s$ are Schur complements given by
\begin{align*}
    \boldS &= \boldV - v^{-1} \Ad^H \boldK_3 \boldi \boldi^H \boldK_3 \Ad, \\
    s &= v - \boldi^H \boldK_3 \Ad \boldV^{-1} \Ad^H \boldK_5 \boldi.
\end{align*}
Observe that
\begin{equation*}
    v
    = \boldi^H \boldW^{-1} \boldi
    = \noisevarp{-4}(M-K) + \boldi^H \boldK_3 \boldi.
\end{equation*}
We know that both $v^{-1}$ and $s^{-1}$ decrease at the rate of $\noisevarp{4}$. As $\noisevar \to 0$, we have
\begin{align*}
    &\boldS \to \Ad^H \boldK_3^\star \Ad, \\
    &s^{-1} (\boldK_5 + \noisevarinv \boldK_1) \to \boldzero, \\
    &v^{-1} (\boldK_5 + \noisevarinv \boldK_1) \to \boldzero, \\
    &s^{-1} (\boldK_5 + \noisevarinv \boldK_1) \boldi \boldi^H
        (\boldK_5  + \noisevarinv \boldK_1) \to \frac{\boldK_1^\star \boldi \boldi^H \boldK_1^\star}{M-K}.
\end{align*}
We now show that $\Ad^H \boldK_3^\star \Ad$ is nonsingular. Denote the eigendecomposition of $\boldA\boldA^H$ by $\Es^\star \Ls^\star (\Es^\star)^H$. Recall that for matrices with proper dimensions, $(\boldA \odot \boldB)^H (\boldC \odot \boldD) = (\boldA^H \boldC) \circ (\boldB^H \boldD)$, where $\circ$ denotes the Hadamard product. We can expand $\Ad^H \boldK_3^\star \Ad$ into
\begin{align*}
    [\boldA^H \Es^\star (\Ls^\star)^{-1} (\Es^\star)^H \boldA]^*
        \circ [\boldA^H \Es^\star (\Ls^\star)^{-1} (\Es^\star)^H \boldA].
\end{align*}
Note that $\boldA\boldA^H \Es^\star (\Ls^\star)^{-1} (\Es^\star)^H \boldA = \Es^\star (\Es^\star)^H \boldA = \boldA$, and that $\boldA$ is full column rank when $K < M$. We thus have $\boldA^H \Es^\star (\Ls^\star)^{-1} (\Es^\star)^H \boldA = \boldI$. Therefore $\Ad^H \boldK_3^\star \Ad = \boldI$, which is nonsingular.

Combining the above results, we obtain that when $\noisevar \to 0$,
\begin{equation*}
    \proj{\boldM_{\bolds}} \to
    \boldK_5^\star \Ad \Ad^H \boldK_5^\star
    + \frac{\boldK_1^\star \boldi \boldi^H \boldK_1^\star}{M-K}.
\end{equation*}
\end{proof}

For sufficiently small $\noisevar > 0$, it is easy to show that $\boldK_1$--$\boldK_5$ are bounded in the sense of Frobenius norm (i.e., $\| \boldK_i \|_F \leq C$ for some $C > 0$, for $i \in \{1,2,3,4,5\}$). Because $\partial\boldr / \partial\boldeta$ is full rank, $\boldM_{\bolds}$ is also full rank for any $\noisevar > 0$, which implies that $\proj{\boldM_{\bolds}}$ is well-defined for any $\noisevar > 0$. Observe that $\proj{\boldM_{\bolds}}$ is positive semidefinite, and that $\trace(\proj{\boldM_{\bolds}}) = \mathrm{rank}(\boldM_{\bolds})$. We know that $\proj{\boldM_{\bolds}}$ is bounded for any $\noisevar > 0$. Therefore $\boldQ_2$ and $\boldQ_3$ are also bounded for sufficiently small $\noisevar$, which implies that $\noisevarsqrt \boldQ_2 + \noisevar \boldQ_3 \to \boldzero$ as $\noisevar \to 0$.

By Lemma~{\ref{lem:proj-ms-exist}, we know that $\boldQ_1 \to \boldQ_1^\star$ as $\noisevar \to 0$, where
\begin{equation*}
    \boldQ_1^\star
    = \DAd^H (\boldK_2^\star - \boldK_4^\star \proj{\boldM_{\bolds}}^\star \boldK_4^\star) \DAd,
\end{equation*}
and $\proj{\boldM_{\bolds}}^\star = \lim_{\noisevar \to 0^+} \proj{\boldM_{\bolds}}$ as derived in Lemma~\ref{lem:proj-ms-exist}. Assume $\boldQ_1^\star$ is nonsingular\footnote{The condition when $\boldQ_1^\star$ is singular is difficult to obtain analytically. In numerical simulations, we have verified that it remains nonsingular for various parameter settings.}. By \eqref{eq:crb-q1-q2-q3} we immediately obtain that $\CRB_{\boldtheta} \to \boldzero$ as $\noisevar \to 0$.

When $K \geq M$, $\boldR$ is full rank regardless of the choice of $\noisevar$. Hence $(\boldR^T \otimes \boldR)^{-1}$ is always full rank. Because $\partial\boldr/\partial\boldeta$ is full column rank, the FIM is positive definite, which implies its Schur complements are also positive definite. Therefore $\CRB_{\boldtheta}$ is positive definite.

\ifCLASSOPTIONcaptionsoff
  \newpage
\fi

\bibliographystyle{IEEEtran}
\bibliography{performance_ref}

\end{document}